\documentclass{article}

\usepackage{etex}

\usepackage{amsmath}
\usepackage{amssymb}
\usepackage{amsthm}
\usepackage{authblk}
\usepackage{dcpic}
\usepackage[hang,flushmargin]{footmisc}
\usepackage[left=1in,right=1in,top=1in,bottom=1in]{geometry}
\usepackage[
	ocgcolorlinks,
	linkcolor=linkblue,
	citecolor=linkred,
	urlcolor=linkred]
{hyperref}
\usepackage{MnSymbol}
\usepackage{pictexwd}
\usepackage{tikz-cd}
\usepackage{verbatim}
\usepackage{xcolor}

\numberwithin{equation}{section}

\definecolor{linkred}{rgb}{0.75,0,0}
\definecolor{linkblue}{rgb}{0,0,0.75}

\theoremstyle{plain}
\newtheorem{maintheorem}{Theorem}

\theoremstyle{plain}
\newtheorem{theorem}{Theorem}[section]
\newtheorem{lemma}[theorem]{Lemma}
\newtheorem{proposition}[theorem]{Proposition}
\newtheorem{corollary}[theorem]{Corollary}

\theoremstyle{definition}
\newtheorem{definition}[theorem]{Definition}
\newtheorem{rem}[theorem]{Remark}

\newcommand{\s}{\sigma}
\newcommand{\B}{\overline{B}}
\renewcommand{\xi}{z}

\begin{document}

\title{The completeness of the Bethe ansatz for the periodic ASEP}

\author{Eric Brattain}
\affil{Department of Mathematics,
SUNY New Paltz,
New Paltz, NY 12561, U.S.A.
\texttt{brattaie@newpaltz.edu}}

\author{Norman Do}
\affil{School of Mathematical Sciences,
Monash University,
VIC 3800, Australia
\texttt{norm.do@monash.edu}}

\author{Axel Saenz}
\affil{Department of Mathematics,
University of Virginia,
Charlottesville, VA 22904, U.S.A.
\texttt{ais6a@virginia.edu}}

\date{\today}

\maketitle

\begin{abstract}
\noindent The asymmetric simple exclusion process (ASEP) for $N$ particles on a ring with $L$ sites may be analyzed using the Bethe ansatz. In this paper, we provide a rigorous proof that the Bethe ansatz is complete for the periodic ASEP. More precisely, we show that for all but finitely many values of the hopping rate, the solutions of the Bethe ansatz equations do indeed yield all $\binom{L}{N}$ eigenstates. The proof follows ideas of Langlands and Saint-Aubin, which draw upon a range of techniques from algebraic geometry, topology and enumerative combinatorics.
\end{abstract}



\setlength{\parskip}{0pt}
\tableofcontents %
\setlength{\parskip}{6pt}


\section{Introduction}

\subsection{Background}

The asymmetric simple exclusion process (ASEP) is a continuous-time Markov process that arose in the study of biopolymers~\cite{pipkin} and was first analyzed mathematically by Spitzer~\cite{spi}. It has since become a significant object of study in the fields of non-equilibrium statistical mechanics and interacting particle systems~\cite{mallick}. The process involves a set of particles located at distinct integral points on the real number line. At various times, a particle may attempt to move one unit to the right with a probability $p$~---~known as the \emph{hopping rate}~---~or one unit to the left with probability $q=1-p$, with the caveat that two particles cannot occupy the same site. The term \emph{asymmetric} refers to the fact that $p$ and $q$ are not necessarily equal; the term \emph{simple} refers to the fact that particles hop one unit; and the term \emph{exclusion} refers to the fact that particles cannot occupy the same site. In this paper, we focus exclusively on the periodic ASEP, in which case the sites lie on a ring. The simple rules governing the ASEP belie profoundly complex dynamics that make it a particularly interesting and useful model to consider.

The Bethe ansatz is a technique that dates back to the seminal paper of Bethe~\cite{Bethe} on the Heisenberg XXX spin chain and has since been applied to a wide variety of models in quantum and statistical mechanics. For the reader looking for the general history and applications of the Bethe ansatz, there are various well-written surveys in the literature~\cite{Baxter,Faddeev, Gaudin, Sutherland}. The version of the ansatz that we study in this paper is often referred to as the coordinate Bethe ansatz, which has been applied to variations of the ASEP model since the work of Gwa and Spohn~\cite{GwaSpohn}. More recently, Tracy and Widom~\cite{TWasep} used formulas inspired by the coordinate Bethe ansatz to show that ASEP on $\mathbb{Z}$ with step initial condition belongs to the KPZ universality class~\cite{Corwin2}.

The issue of the completeness of the Bethe ansatz~---~in other words, whether it yields all possible eigenstates~---~has been addressed rigorously in only a handful of cases. Determining completeness is not only appealing from the standpoint of mathematical rigour, but is also desirable for various other reasons, such as the asymptotic analysis of ASEP starting from finite models. In 1994, Dorlas~\cite{Dorlas} proved completeness for the nonlinear Schr\"{o}dinger model but stated: ``In the Heisenberg chain the existence of bound states presents a problem in finite volume." 
Indeed, for the ASEP and Heisenberg XXZ models on the 1-dimensional lattice $\mathbb{Z}$, completeness was only recently proven by Borodin, Corwin, Petrov and Sasamoto~\cite{Borodin} as a corollary to powerful new algebraic insights into the structure of integrable systems. However, the finite periodic case has remained immune to these methods.

In the 1990s, Langlands and Saint-Aubin~\cite{Langlands,FrenchLanglands} outlined an argument based upon classical results in algebraic geometry and topology that reduces the problem of completeness of the Bethe ansatz for the XXZ model to a certain enumeration of spanning forests. In this paper, we expand on the argument of Langlands and Saint-Aubin and apply it to the case of the periodic ASEP. We hope to make these powerful techniques more well-known and accessible to a wider audience. Our main result is the following.

\begin{maintheorem} \label{mainthm}
For generic values of the hopping rate $p$, the Bethe ansatz for the periodic ASEP is complete.\end{maintheorem}

In this theorem, we interpret $p$ as a point in the Riemann sphere, rather than a real number in the interval $[0,1]$, in order to employ tools from algebraic geometry and topology. Furthermore, we use the term \emph{generic} to mean that $p$ lies in a Zariski open subset of $\mathbb{CP}^1$. In other words, the theorem is true for all but finitely many values of $p$. At such non-generic values of the hopping rate $p$, the Bethe ansatz does not yield a complete basis of eigenstates. However, a key observation is that the generic completeness guaranteed by Theorem~\ref{mainthm} allows one to use a limiting procedure to obtain a complete eigenspace decomposition for any value of $p$. In general, the basis given by this limiting procedure does not diagonalize the operator, but  decomposes it into Jordan canonical form.

\subsection{ASEP and the Bethe ansatz}

The ASEP model on the integer lattice $\mathbb{Z}$ is a continuous-time Markov process $\eta_t$, where we set $\eta_t(x) = 1$ if the site $x \in \mathbb{Z}$ is occupied by a particle at time $t$ and $\eta_t(x) = 0$ if the site $x \in \mathbb{Z}$ is vacant at time $t$~\cite{Liggett}. Each particle is controlled by an exponential timer with parameter 1 such that, when the timer activates, the particle attempts to jump from its current site at $x$ to an adjacent site $y = x \pm 1$ with a certain probability $p(x,y)$. If $y = x \pm 1$ is occupied, then $p(x, y) = 0$; if $y = x+1$ is unoccupied, then $p(x,y) = p$; and if $y = x-1$ is unoccupied, then $p(x,y) = q = 1-p$. In this paper, we focus exclusively on the ASEP model with periodic boundary conditions, in which each particle occupies one of $L$ sites on a ring. In other words, we work on $\mathbb{Z}_L$ rather than $\mathbb{Z}$. 

For a positive integer $n$, we introduce the notation
\[
[n] = \{1, 2, \ldots, n\} \qquad\qquad \text{and} \qquad \qquad [n]^{(2)} = \{ (i, j) \in \mathbb{Z}^2 \mid 1 \leq i < j \leq n\}.
\]
Let $\mathcal{X} = \displaystyle \bigotimes_{i=1}^L \mathbb{C}^2_i$ and consider the basis $\{u_{S} \mid S \subseteq [L] \}$, where $\mathbb{C}^2_i = \langle u_i^{+} , u_i^{-} \rangle$ and
\[
u_S = \bigg( \bigotimes_{i \in S} u_i^{+}\bigg) \otimes \bigg( \bigotimes_{i \in [L] \setminus S} u_i^{-}\bigg).
\]
The basis element $u_S$ naturally represents the state in which the sites in $S$ are occupied while the remaining sites are empty. Define the operator
\[
H(t;p) : \mathcal{X} \rightarrow \mathcal{X}, 
\]
where 
\[
 \sum_{S \subseteq [L]} a^{S}u_S \mapsto \sum_{S \subseteq [L]} b^{S}u_S
\]
and
\[
b^{S} = \sum_{S'} p( a^{S} - a^{S'}) + \sum_{S''} q(a^{S} - a^{S''}).
\]
Here, the summations range over all states $S'$ that differ from $S$ by a single particle that moved one position to the right and states $S''$ that differ from $S$ by a single particle that moved one position to the left. Note that $|S| = |S'| = |S''|$, so that the number of particles is conserved. Thus, the operator $H(t;p)$ decomposes into a direct sum of operators
\begin{equation}\label{hamiltonian}
H_N(t;p) : \mathcal{X}_N \rightarrow \mathcal{X}_N,
\end{equation}
where $\mathcal{X}_N = \mathrm{span} \{ u_{S} \mid |S| = N\}$ is the subspace with dimension $\binom{L}{N}$ whose basis elements correspond to states with $N$ particles.

One can solve the periodic ASEP model by diagonalizing the operator $H_N(t;p)$ using the Bethe ansatz, as was carried out by Gwa and Spohn~\cite{GwaSpohn}. 
Let $x_1, x_2, \ldots, x_N$ denote the sites occupied by the $N$ particles. The probability $u(X; t)$ of being in state $X = (x_1, \ldots, x_N)$ at time $t$ satisfies the so-called \emph{master equation}
\begin{equation} \label{master}
\frac{\partial u}{\partial t} = \sum_{i=1}^N [p u(x_i - 1) \delta_{i, i-1} + q u(x_i + 1)\delta_{i+1, i} - p u(x_i) \delta_{i+1,i } - q u(x_i) \delta_{i, i-1}].
\end{equation} 
Here, we use the notation $u(x_i \pm 1)$ as a shorthand for $u(x_1, \ldots, x_{i-1}, x_i \pm 1, x_{i+1}, \ldots x_N; t)$ and we define $\delta_{i,j} = 0$ if $x_i = x_j+1$ and $\delta_{i,j} = 1$ otherwise. The periodic boundary conditions impose the constraint $u(x_1,\ldots,x_N; t) = u(x_2,\ldots, x_N, x_1+L; t)$ and the initial state $Y = (y_1, \ldots, y_N)$ imposes the constraint $u(x_1,\ldots,x_N; 0) = \delta_{X,Y}$. We are not concerned here with the full computation of $u(X; t)$, so we shall not mention initial conditions again. 

The coordinate Bethe ansatz is derived from the observation that the master equation~\eqref{master} greatly simplifies if one assumes that the particles are sufficiently far apart from each other to be considered non-interacting. In this case, one obtains the so-called \emph{free equation}
\begin{equation} \label{freemaster}
\frac{\partial u}{\partial t} = \sum_{i=1}^N [p u(x_i - 1) + qu(x_i + 1) - u(x_i)].
\end{equation}
This free equation becomes equivalent to the master equation~\eqref{master} once we impose the boundary conditions
\begin{equation} \label{bcs}
p u(x_i,x_i) + q u(x_i + 1, x_i + 1) - u(x_i, x_i + 1) = 0.
\end{equation}

The Bethe ansatz proposes solutions of the form
\begin{equation} \label{ansatz}
u_{\vec{\xi}}(x_1, \ldots, x_N) = \sum_{\sigma \in S_N}A_{\sigma} \prod_{i=1}^N \xi_{\sigma(i)}^{x_i},
\end{equation}
where $\vec{z} = (z_1, \ldots, z_N)$ is a $N$-tuple of complex parameters and, for each $\sigma$,  $A_{\sigma}$ is a function of $z_1, \ldots, z_N$ chosen to satisfy the boundary conditions of equation~\eqref{bcs}.
As a result, one obtains
\begin{equation} \label{amplitudes}
A_{\s} (z_1, \ldots, z_N) = \mathop{\prod_{\textrm{inversions}}}_{(i,j)} - \frac{p + q \xi_i \xi_j - \xi_i}{p + q \xi_i \xi_j - \xi_j},
\end{equation}
where we recall that an \emph{inversion} of a permutation $\s$ is a pair $(i,j)$ such that $i < j$ and $\s(i) > \s(j)$.

Now imposing the periodicity constraint $u(x_1,\ldots,x_N; t) = u(x_2,\ldots,x_N, x_1 + L; t)$ on the ansatz of equation~\eqref{ansatz} leads to the \emph{Bethe ansatz equations}
\begin{align} \label{BetheEquations}
\xi_j^L &= (-1)^{N-1}\prod_{i=1}^N \frac{p + q \xi_j \xi_i - \xi_j}{p + q \xi_j \xi_i - \xi_i} \qquad \text{for } j=1, 2, \ldots, N.
\end{align}
Note that these equations imply the condition $\prod_{i=1}^N \xi_i^L = 1$. 

In this paper, we count the number of solutions to these equations and thereby determine the number of eigenstates produced by the Bethe ansatz. The structure of the paper is as follows.
\begin{itemize}
\item In Section 2, we analyze the case of $N = 2$ particles in detail. We introduce the algebro-geometric perspective and topological tools --- in particular, the Lefschetz theorem --- required to approach the general case.
\item In Section 3, we consider the geometric set-up of the general case of $N$ particles. Inadmissibility conditions are introduced to classify those solutions of the Bethe ansatz equations that do not actually lead to eigenstates. We show that solutions of the Bethe ansatz equations correspond to coincidences between certain algebraic functions $\psi$ and $\phi$. Such coincidences can then be counted by the Lefschetz theorem.
\item In Section 4, we compute the traces appearing in the Lefschetz theorem and show how they relate to the enumeration of trees. By a careful combinatorial argument, we deduce that the number of admissible solutions to the Bethe ansatz equations is $L(L-1) \cdots (L-N+1)$, as required.
\item In Section 5, we consider non-generic values of the hopping rate $p$, at which the Bethe ansatz fails to describe all possible eigenstates. We argue that a limiting procedure can always be used in conjunction with the Bethe ansatz to express the operator in Jordan canonical form.
\item In Appendix A, we discuss the algebro-geometric notion of blow-ups in the context of the maps $\psi$ and $\phi$ defined earlier. Whereas these are rational maps and hence, only defined on a Zariski open subset of their domains, the Lefschetz theorem requires a continuous maps between compact spaces. We therefore blow up the domains appropriately and prove the existence of smooth resolutions of the maps $\psi$ and $\phi$.
\end{itemize}


\section{The two-particle case} \label{sec:2particles}

In this section, we examine the case of two particles on a ring with $L$ sites. This simple example illustrates how algebraic geometry, topology and combinatorics come into play. In later sections, these ideas are combined to prove completeness of the Bethe ansatz for the general case of $N$ particles on a ring with $L$ sites.

\subsection{Naive approach}

In the case $N = 2$, the Bethe ansatz equations~\eqref{BetheEquations} reduce to the following, where $p + q = 1$.
\begin{equation} \label{betheqns2}
\xi_1^L = - \frac{p + q \xi_1 \xi_2 - \xi_1}{p + q \xi_1 \xi_2 - \xi_2} \qquad \qquad \qquad \qquad \xi_2^L = - \frac{p + q \xi_1 \xi_2 - \xi_2}{p + q \xi_1 \xi_2 - \xi_1}
\end{equation}
The Bethe ansatz then asserts that, for $z_1$ and $z_2$ satisfying these equations, one should seek eigenstates of the operator $H_2$ (see eqn.~\eqref{hamiltonian}) given by
\[
u(x_1, x_2) = A_{12} z_1^{x_1} z_2^{x_2} + A_{21} z_1^{x_2} z_2^{x_1},
\]
where the amplitudes $A_{12}$ and $A_{21}$ are thus related by equation~\eqref{amplitudes}.
\[
A_{12} = - A_{21} \frac{p + q \xi_1 \xi_2 - \xi_1}{p + q \xi_1 \xi_2 - \xi_2}
\]

However, note that simultaneous solutions to equations~\eqref{betheqns2} do not necessarily yield non-trivial eigenstates. First, if $z_1 = 0$ or $z_2 = 0$, then it is clear that $u(x_1, x_2) = 0$ for all $x_1$ and $x_2$. Second, if $z_1 = z_2 \neq 1, p/q$, then we have $A_{12} = -A_{21}$ and $u(x_1, x_2) = 0$ for all $x_1$ and $x_2$. Therefore, we declare solutions to \eqref{betheqns2} such that $u(x_1, x_2) =0$ for all $x_1$ and $x_2$ (such as $z_1 = z_2 \neq 1, p/q$) to be \emph{inadmissible}. However if $z_1 = z_2 = 1, p/q$, the Bethe equations ~\eqref{BetheEquations} and the amplitudes~\eqref{amplitudes}  are not well-defined, and we note that the steady state solution of ASEP corresponds to the case when $z_1 =z_2 =1$. We have to deal with these points where the Bethe equations are not well-defined in a special way. In the general case, we will keep encountering points where the Bethe equations are not well-defined, which we resolve by applying blow-ups to the domain of the functions (see sec.~\ref{sec:maps} and Appendix~\ref{blow}). In this base case of two particles, we can resolve the Bethe equations by clearing out the denominator from the rational functions, which works well for this example but in the general case we rather keep rational functions and apply blow-ups. Lastly, one can check that \emph{admissible} solutions~---~in other words, those with distinct $z_1, z_2 \in \mathbb{C}^*$ away from the not well-defined locus~---~do indeed yield non-trivial eigenstates. Furthermore, two admissible solutions $(z_1, z_2)$ and $(z_1', z_2')$ yield the same eigenstate if and only if they are permutations of each other. Therefore, we aim to show that there are $L(L-1)$ admissible solutions to equations~\eqref{betheqns2}, from which we may deduce that there are $\frac{1}{2} L(L-1) = \binom{L}{2}$ distinct eigenstates.


Let us initially consider a naive approach to the problem, in order to better understand the nature of the inadmissible solutions. 
Rewrite the system of equations~\eqref{betheqns2} as
\begin{equation} \label{epsilonbethe}
\xi_2 = \epsilon \xi_1^{-1} \text{ with } \epsilon^{L} = 1 \qquad\qquad \text{and} \qquad\qquad
\xi_1^L = - \frac{p + q \epsilon - \xi_1}{p + q \epsilon - \epsilon \xi_1^{-1}}.
\end{equation}
At first glance, it appears as though there are $L^2$ solutions, since there are $L$ choices for $\epsilon$ and each one yields $L$ solutions for $z_1$. However, it is necessary to exclude the inadmissible solutions with $\xi_1 = \xi_2 \neq 1 , p/q$. In this case, equations~\eqref{epsilonbethe} reduce to
\[
\epsilon \xi_1^{-1} = \xi_1 \text{ with } \epsilon^L = 1 \qquad \qquad \text{and} \qquad \qquad \xi_1^L = -1.
\]
Therefore, for each of the $L$ choices for $\epsilon$, there is exactly one solution $z_1 = \pm \epsilon^{1/2}$ that is compatible with the equation $z_1^L = -1$. So we obtain $L$ inadmissible solutions with $z_1 = z_2$. Note that inadmissible solutions with $z_1 = 0$ or $z_2 = 0$ only arise when $p = 0$, which we presently exclude from our consideration.

For the non-well-defined points $z_1 = z_2 = 1, p/q$, we clear out the denominator on the left equation of \eqref{epsilonbethe}. Then, the left equation of \eqref{epsilonbethe} becomes trivial for the non-well-defined points of the Bethe equations but we still need to satisfy the right equation of \eqref{epsilonbethe}. We must have that $\epsilon = (1)^2 ,(p/q)^2$ with $\epsilon^L = 1$. For the ASEP, we have that $p+q=1$ with $p,q \in \mathbb{R}$. So, when we resolve the Bethe equations for the non-well-defined points, we obtain only one more solution $z_1 =z_2 = 1$ (in the physical case of ASEP) corresponding to the steady state solution. Moreover, one should note that for the special solution of $z_1 =z_2 =1$, the amplitudes $A_{\sigma}$ are not well-defined and one can determine by inspection that $A_{12} =A_{21}=1$ for this case. In general, one may define the value of the amplitudes for the non-well-defined points after these equations are resolved via the blow-ups.

Note that, clearing the denominator of~\eqref{epsilonbethe} and obtaining a polynomial in $\xi_1$ and $p$ ($q = 1-p$), the inadmissible solution corresponds to factoring the polynomial
\begin{equation}\label{factor}
(p+ q\epsilon) \xi_1^L - \epsilon \xi_1^{L-1} - \xi_1 +(p+ q\epsilon) = (\xi_1 \pm \epsilon^{1/2}) f(\xi_1, p).
\end{equation}
Thus, we have that for each $\epsilon$, and for every $p$, there is is an inadmissible solution, and the rest of the solutions, which number $L(L-1)$ as we claimed, should be admissible solutions.


\subsection{Algebro-geometric perspective}

We now approach the problem of determining the number of admissible solutions using a more sophisticated perspective. Solutions to equations~\eqref{betheqns2} correspond to zeros of the ideal
\[
I = \big( (p + q \xi_1\xi_2 - z_2) \xi_1^L + (p + q \xi_1 \xi_2 - z_1), (p+ q\xi_1 \xi_2 - z_1) \xi_2^L + (p+ q\xi_ 1\xi_2  - z_2) \big), 
\]
whose generators are obtained by clearing the denominators in equations~\eqref{betheqns2}. There is a small subtlety in clearing the denominators, since one may inadvertently introduce new solutions in the case that $p+qz_1z_2-z_1 = p+qz_1z_2-z_2 = 0$. However, these equations imply that $z_1 = z_2= 1, p/q$, which yields an inadmissible solution that will eventually be excluded. Since these generators are polynomials in $z_1$, $z_2$ and $p$, one might expect to be able to bring the tools of algebraic geometry to bear. Thus, we consider the following commutative diagram of natural ring homomorphisms.
\[
\begin{tikzcd}
\mathbb{C}[z_1,z_2,p]/I \arrow[leftarrow]{r}{} \arrow[leftarrow]{dr}{}& \mathbb{C}[z_1,z_2,p] \arrow[leftarrow]{d}{}\\
& \mathbb{C}[p]
\end{tikzcd}
\]
Algebro-geometrically, this corresponds to the following morphisms between varieties.
\[
\begin{tikzcd}
\Sigma := \text{Spec }\mathbb{C}[z_1,z_2,p]/I  \arrow{r}{} \arrow{dr}{\pi}& \text{Spec } \mathbb{C}[z_1,z_2,p] = \mathbb{C}^3 \arrow{d}{} \arrow[hook]{r}{} & \left(  \mathbb{CP}^1 \right)^3 \arrow{d}{} \\
& \text{Spec }\mathbb{C}[p] = \mathbb{C} \arrow[hook]{r}{} & \mathbb{CP}^1
\end{tikzcd}
\] 
Note that we have extended the commutative diagram to include the natural embeddings $\mathbb{C} \to \mathbb{CP}^1$ and $\mathbb{C}^3 \to (\mathbb{CP}^1)^3$. This will play a role later on, when we introduce topological tools that require us to work with compact spaces.


Moreover, it follows easily from basic results in algebraic geometry (cf.~\cite{Hartshorne} Prop. 1.13 p.7) that the space $\Sigma$ is a 1-dimensional complex manifold (i.e. an orientable real surface with a complex structure). Thus, in the case of two particles, the map $\pi : \Sigma \rightarrow \mathbb{C}$ of surfaces is generically $L^2$-to-1, and by factoring the inadmissible solutions, we have shown that $\Sigma$ decomposes into a union of independent components 
\[
\Sigma = \Sigma_{\text{ad}} \cup \Sigma_{\text{in}}
\]
such that $\pi$ factors into maps
\[
\Sigma_{\text{ad}} \overset{L(L-1) : 1}{\longrightarrow} \mathbb{C} \qquad \text{and} \qquad \Sigma_{\text{in}} \overset{L : 1}{\longrightarrow} \mathbb{C},
\]
where the first map is the map of admissible solutions and the second map is the map of inadmissible solutions. These observations follow from the factorization~\eqref{factor} where the roots found correspond to the inadmissible solutions.

In the general case of $N$ particles, we wish to define a surface $\Sigma$ and a map $\pi : \Sigma \rightarrow \mathbb{C}$ via the Bethe ansatz equations. Then we decompose $\Sigma$ into the space of admissible solutions, $\Sigma_{\text{ad}}$, and the space of inadmissible solutions, $\Sigma_{\text{in}}$, and show that the induced map 
\[
\Sigma_{\text{ad}} \rightarrow \mathbb{C}
\]
is generically $N! \binom{L}{N}$ to 1.

This will be accomplished by counting fixed points (correspondences) via the Lefschetz theorem~\cite{Lefschetz} and using the combinatorics of rooted trees to get the proper count. Still, even though the maps $\Sigma_{\text{ad}} \overset{\pi}{\rightarrow} \mathbb{C}$ are generically (i.e. for all but finite $p$) $N! \binom{L}{N}$ to 1, there will be special points $p_r \in \mathbb{C}$, called ramification points, such that $|\pi^{-1}(p_r)| < N! \binom{L}{N}$. For these points, the Bethe ansatz might not be complete, but knowing that $\Sigma_{\text{ad}}$ is independent from $\Sigma_{\text{in}}$, we can perform a limiting procedure that will transform $H_N(t;p_r)$ into Jordan canonical form. We develop this further in Section~\ref{limiting}.

We have a map of surfaces $\pi: \Sigma \rightarrow \mathbb{C}$, where the fiber $\pi^{-1}(p)$ for each $p \in \mathbb{C}$ consists of solutions to the Bethe ansatz equations, both admissible and otherwise. In the previous section, we  found $L$ inadmissible solutions. We now perform the count of admissible solutions using topological tools that will allow us to generalize to the case of $N$ particles.

Start by fixing a generic value of $p \in \mathbb{C}$ and consider the rational map $\phi : (\mathbb{CP}^1)^3 \dashedrightarrow (\mathbb{CP}^1)^3$ defined by
\[
([\xi_0^1 : \xi_1^1], [\xi_0^2 : \xi_1^2], [\omega_0 : \omega_1]) \mapsto 
([\omega_0 : \omega_1], [\omega_1 : \omega_0], [p \xi_0^1\xi_0^2 + q \xi_1^1 \xi_1^2 - \xi_1^1 \xi_0^2 : -(p \xi_0^1 \xi_0^2 + q \xi_1^1 \xi_1^2 - \xi_1^1 \xi_0^2)]).
\]
It is well-defined 
away from the codimension two subvariety
\[
Z := \bigcup_{i=0,1} \left\{ ([q : q] , [q : p] ) \right\} \times \mathbb{CP}^1 \subseteq (\mathbb{CP}^1)^3.
\]
Furthermore, consider the map $\psi : (\mathbb{CP}^1)^3 \rightarrow (\mathbb{CP}^1)^3$ defined by
\[
([\xi_0^1: \xi_1^1],[\xi_0^2 : \xi_1^2], [\omega_0: \omega_1]) \mapsto ([(\xi_0^1)^L: (\xi_1^1)^L],[(\xi_0^2)^L : (\xi_1^2)^L], [\omega_0: \omega_1]).
\]
Then, we choose a chart $U = (\xi_0^1\neq 0, \xi_0^2 \neq 0, \omega_0 \neq 0 )$ and $f: U \rightarrow \mathbb{C}^3$ given by 
\[
([\xi_0^1: \xi_1^1],[\xi_0^2 : \xi_1^2], [\omega_0: \omega_1]) \mapsto (\xi_1^1/ \xi_0^1,\xi_1^2 / \xi_0^2, \omega_1/ \omega_0)
\]%
Moreover, on this chart and away from $Z$ we have that the maps $f \circ \psi \circ f^{-1} , f \circ \phi \circ f^{-1} : \mathbb{C}^3 \dashedrightarrow \mathbb{C}^3$ are given by 
\begin{align*}
f \circ \psi \circ f^{-1}(\xi_1, \xi_2, w_{12}) &= (\xi_1^{L} , \xi_2^{L}, w_{12}) \\
f \circ \phi \circ f^{-1}(\xi_1, \xi_2, w_{12}) &= \left(w_{12}, w_{12}^{-1}, - \frac{p +q\xi_1\xi_2 - \xi_1}{p +q\xi_1\xi_2 - \xi_2}\right).
\end{align*}

Now, in this chart, we note that the system of equations~\eqref{betheqns2} is equivalent to the equation
\begin{equation} \label{corr1}
f \circ \psi \circ f^{-1} (\xi_1, \xi_2, w_{12}) = f \circ \phi \circ f^{-1} (\xi_1, \xi_2, w_{12}).
\end{equation}

We wish to count the number of solutions to~\eqref{corr1} by applying the Lefschetz theorem~\ref{lefschet}. However, this result applies only to continuous maps between compact spaces. Therefore, it is necessary to  resolve $\psi: (\mathbb{CP}^1)^3 \dashedrightarrow (\mathbb{CP}^1)^3$ into a smooth map. We do this by blowing up the domain along the subvariety $Z$ and compactifying the codomain. We explain the details of this construction for the general case in Appendix~\ref{blow}. Then, we define 
\begin{align*}
C &:= \mathrm{Blow}_{Z}((\mathbb{CP}^1)^3)\\
X &:= (\mathbb{CP}^1)^3
\end{align*}
There exist smooth maps $\tilde{\psi} , \tilde{\phi} : C \rightarrow X$ that extend the maps $\psi, \phi : X \dashedrightarrow X$ on the subvariety $Z$. Now, with these smooth resolutions, we can apply the Lefschetz theorem~\ref{lefschet}. Actually, we will use the Lefschetz theorem to count the number of solutions of 
\begin{equation} \label{corr2}
\tilde{\psi}(pt) = \tilde{\phi}(pt)
\end{equation}%
on the whole domain $C$ rather than just the chart $U$. Of course, this way we will inadvertently introduce extra solutions outside the chart $U$ that don't correspond to an eigenstate, but we will classify and discard those solutions along the way.

\subsection{Topological tools}

We now present the topological technique that we use to compute the number of points in $X$ that satisfy~\eqref{corr2}, denoted by $\lambda (\tilde{\psi} , \tilde{\phi})$.

\begin{theorem}[Strong Lefschetz theorem] \label{lefschet}
Given two differentiable maps $\phi , \psi : X \rightarrow Y$ of compact spaces the number of solutions (with multiplicity) of the equation $\phi (x) = \psi(x) $, the coincidence number of $\phi$ and $\psi$, is given by
\begin{equation} \label{trace}
\lambda  (\psi , \phi) := \sum_{i=0}^{\dim Y} (-1)^i \, \mathrm{Tr}(\psi_i \phi^i),
\end{equation}
where $\psi_i: H_i (X) \rightarrow H_i(Y)$ is the pushforward in homology and $\phi^i : H_i (Y) \rightarrow H_i(X)$ is the Poincar\'{e} dual of the pullback in cohomology. We also call $\lambda(\psi,\phi)$ the Lefschetz number.
\end{theorem}

This theorem, in this generality, is not found in many textbooks, but it is discussed in old textbooks by Lefschetz himself, such as in~\cite{Lefschetz}. It is a quick check to see that our formulas satisfy the hypothesis of the theorem. As a matter of fact, this is the reason that we need to apply blow-ups to our original domain. Otherwise, we would have a singular domain and the theorem would not apply. Thus, we can use this formula to find the number of solutions of equation~\eqref{correspondence}.

\begin{rem}
In Theorem~\ref{lefschet}, we use homology with $\mathbb{C}$-coefficients. This a choice made by the authors for convenience in the proofs to follow. In any case, we will use $\mathbb{C}$ throughout the paper. So, we will denote $H_*(X;\mathbb{C})$ and $H_*(C;\mathbb{C})$ by $H_*(X)$ and $H_*(C)$, accordingly.
\end{rem}

We compute the trace of the induced maps on the cohomology vector spaces
\begin{align*}
\psi_i: H_i(C) &\rightarrow H_i(X) \\
\phi^i :H_i(X) &\rightarrow H_i(C)
\end{align*}
where the last map is defined via Poincar\'{e} duality.

Since $X = (\mathbb{CP}^1)^3$ and $C = \mathrm{Blow}_Z(X)$, we have
\begin{align*}
H_{*}(X) &= H_0(X) \oplus H_2(X) \oplus H_4(X) \oplus H_6(X) \\
&= \langle 1\rangle \oplus \langle e_1, e_2, e_{12}\rangle \oplus \langle e_1\otimes e_2, e_1 \otimes e_{12}, e_2 \otimes e_{12}\rangle \oplus \langle e_1 \otimes e_2 \otimes e_{12}\rangle \\
H_*(C) &= H_*(X) \oplus H.
\end{align*}
Here, the second summand comes from the blow-up (cf. Prop.~\ref{directsum} ). One can show that $\phi^i(H_*(X)) \subseteq H_*(X) $, making the direct summand $H$ of $H_*(C)$ irrelevant to our computations, and so we ignore it for the remainder. This is explained in more detail and  more generality in Appendix~\ref{blow} and Proposition~\ref{directsum}. 

So the induced maps are
\[
\left( \bigotimes_{i \in S } e_i \right) \otimes \left( \bigotimes_{t \in T} e_t \right) \overset{\phi^*}{\mapsto} L^{|\{1,2\} \setminus S |} \left( \bigotimes_{i \in S } e_i \right) \otimes \left( \bigotimes_{t \in T} e_t \right)
\]%
where $ S \subseteq \{ 1,2\}$ and $T \subseteq \{ (1,2) \}$.

The induced linear map $\psi_{*}$ is determined in the appendix for the general case. For two particles, we can simply compute it directly and find 
\begin{align*}
1 &\overset{\psi_{*}}{\mapsto} 1 & e_1 \otimes e_2 &\overset{\psi_*}{\mapsto} 0 \\
e_1, e_2 &\overset{\psi_{*}}{\mapsto} e_{12} & e_1 \otimes e_{12}, e_2 \otimes e_{12} &\overset{\psi_{*}}{\mapsto} e_1 \otimes e_{12} + e_2 \otimes e_{12} \\
e_{12} &\overset{\psi_{*}}{\mapsto} e_1 + e_2 & e_1 \otimes e_2 \otimes e_{12} &\overset{\psi_{*}}{\mapsto} 0.
\end{align*}
So, we compute the trace and obtain 
\[
\lambda (\psi,\phi) = L^2 +2 L,
\]
which counts all of the solutions, admissible or otherwise. Subtracting the $L$ inadmissible solutions as we did in the elementary algebraic approach, it still seems as though we have $2L$ too many. These extra $2L$ solutions are the solutions at infinity that we have inadvertently added by compactifying our spaces. So, we now have another class of inadmissible solution: $\xi_1 = \xi_2$ or $\xi_i= 0, \infty$. The second set of inadmissible solutions really corresponds to the zeros and poles of $\xi_i^{L}$, and having the poles and zeros at infinity with non-trivial multiplicity complicates the problem unnecessarily. Thus, following Langlands and Saint-Aubin~\cite{Langlands}, we replace the polynomial map $\xi_i \mapsto \xi_i^L$ with the rational map
\[
\xi_i \mapsto R(\xi_i) = \prod_{i=1}^L \frac{z-\alpha_i}{z-\beta_i},
\]
where $\alpha_1, \alpha_2, \ldots, \alpha_L$ are distinct and
\[
\beta_i = \frac{1}{q} - \frac{p}{q \alpha_i} \qquad \text{for } i = 1, 2, \ldots, L.
\]
This new function is homotopy equivalent to $\xi_i^L$ and it will thus give us the same results in the computation of the trace of the Lefschetz numbers, and the new inadmissibility conditions are: $\xi_1 = \xi_2$, $\xi_i = \alpha_j$, and $\xi_i = \beta_j$. Moreover, we note that under the deformation of the roots and poles of $R(z)$ the steady state solution of $z_1 = \cdots = z_N$ moves away from the locus where $z_i =z_j$. In particular, we have that $R(1) = (q/p)^L \prod_{i=1}^N \alpha_i$ which is not consistent with the condition that $\prod_{j=1}^N R(z_j) = 1$, meaning that $z_1 = \cdots = z_N =1$ is not a solution of the Bethe equations for a generic choice of $R(z)$.

For each inadmissibility condition, we consider the set of solutions corresponding to it by defining a closed subvariety that corresponds to such a condition and counting the solutions via the Lefschetz Theorem~\ref{lefschet} applied to this subvariety. For example, for the condition $\xi_1 = \xi_2$, we define
\begin{align*}
X\langle \xi_1 = \xi_2 \rangle & : = \{ (z_1, z_2, w_{12}) \in X \mid \xi_1 = \xi_2 \} \subseteq X \\
C\langle \xi_1 = \xi_2 \rangle & : = \{ (z_1, z_2, w_{12}) \in C \mid \xi_1 = \xi_2 \} \subseteq C .
\end{align*}

Then, we have
\[
\lambda (   X \langle \xi_1 = \xi_2 \rangle ) = L
\]

where we have abused notation since $\psi$ and $\phi$ are fixed and the subvarites are the main diference. Similarly, we define closed subspaces for other inadmissibility conditions and compute their Lefschetz numbers.
\begin{align*}
\lambda (X \langle \xi_1 = \alpha_i \rangle ) &= 1 & \lambda (X \langle \xi_1 = \alpha_i, \xi_2 = \alpha_j \rangle ) &= 0 \\
\lambda (X \langle \xi_1 = \beta_i \rangle ) &= 1 & \lambda (X \langle \xi_1 = \alpha_i, \xi_2 = \beta_j \rangle ) &= \delta_{i,j} \\
\lambda (X \langle \xi_2 = \alpha_i \rangle ) &= 1 & \lambda (X \langle \xi_1 = \beta_i, \xi_2 = \alpha_j \rangle ) &= \delta_{i,j} \\
\lambda (X \langle \xi_2 = \beta_i \rangle ) &= 1 & \lambda (X \langle \xi_1 = \beta_i, \xi_2 = \beta_j \rangle ) &= 0.
\end{align*}

Then, the number of inadmissible solutions, via the Inclusion-Exclusion principle of set theory, is
\begin{align*}
& \lambda (X \langle \xi_1 = \xi_2\rangle ) + \sum_i \left[ \lambda (X \langle \xi_1 = \alpha_i \rangle ) + \lambda (X \langle \xi_1 = \beta_i\rangle ) + \lambda (X \langle \xi_2 = \alpha_i \rangle ) + \lambda(X \langle \xi_2 = \beta_i \rangle) \right] \\
& - \sum_{i,j} \left[\lambda (X \langle \xi_1 = \alpha_i, \xi_2 = \alpha_j\rangle ) + \lambda (X \langle \xi_1 = \alpha_i, \xi_2 = \beta_j \rangle ) + \lambda (X \langle \xi_1 = \beta_i, \xi_2 = \alpha_j \rangle ) + \lambda (X \langle \xi_1 = \beta_i, \xi_2 = \beta_j \rangle ) \right],
\end{align*}
which computes to $3L$. That is, the number of admissible solutions is $L^2 +2L -3L = L^2 - L$ just as we showed in the previous section. This process generalizes completely, and it becomes a combinatorial problem to compute the traces of the homology maps on the closed subspaces of the inadmissible solutions.

\section{The geometric set-up}

In this section, we discuss the geometric set-up for the general case of the Bethe ansatz for the periodic ASEP with $N$ particles.

\subsection{Inadmissibility conditions}

Recall that the Bethe ansatz equations~\eqref{BetheEquations} state that
\begin{align} \label{L}
\xi_j^L &= (-1)^{N-1}\prod_{i=1}^N \frac{p + q \xi_j \xi_i - \xi_j}{p + q \xi_j \xi_i - \xi_i}, \qquad \text{for $j=1, 2, \ldots, N$}. 
\end{align}
Clearing the denominators, we define the polynomials
\[
f_j(\vec{z}, p) = \xi_j^L \prod_{i=1}^N \left( p +q \xi_{j} \xi_{i} - \xi_{i} \right) + (-1)^N \prod_{i=1}^N \left( p + q \xi_{j} \xi_{i} - \xi_{j} \right), \qquad \text{for } j = 1, 2, \ldots, N.
\]
Define the ideal $I = (f_1, f_2, \ldots , f_N) \subseteq \mathbb{C}[\vec{\xi}, p]$ and consider the inclusion and projection maps
\[
\mathbb{C}[p] \to \mathbb{C}[\vec{\xi}, p] \rightarrow \mathbb{C}[\vec{\xi}, p] / I.
\]
By applying the contravariant functor $\mathrm{Spec}$, one obtains a morphism of varieties
\[
\Sigma' \to \mathbb{C}^{N+1} \to \mathbb{C},
\]
with $\Sigma$ the closure of $\Sigma '$ in $(\mathbb{CP}^1)^{N+1}$. By considering $\Sigma' \subseteq \mathbb{C}^{N+1} \subseteq (\mathbb{CP}^1)^{N+1}$ and $\mathbb{C} \subseteq \mathbb{CP}^1$, we may extend this composition to a map between compact spaces
\[
\Sigma \overset{\pi}{\rightarrow} \mathbb{CP}^1.
\]
We then have that the preimage of a point $\pi^{-1}(p)$ is the set of solutions of the Bethe ansatz. It turns out that many of the preimage solutions lead to trivial eigenfunctions. In order to give a proper count, we must discard the solutions that lead to trivial eigenfunctions. So, we have the following definition.

\begin{definition}
We say that a solution $\vec{\xi}$ to the Bethe ansatz equations is \emph{inadmissible} if the resulting eigenstate satisfies $u_{\vec{\xi}}(\vec{x})  = 0$ for all $\vec{x}$.
\end{definition}

Inadmissible solutions admit a simple description via the following result, which generalizes the observations of Section~\ref{sec:2particles} in the case of two particles.

\begin{proposition}
\label{inad}
Suppose that $\vec{\xi}$ is a solution to the Bethe ansatz equations such that $z_i \in \mathbb{C}^*$ for $i = 1, 2, \ldots, N$, $(q/p)^{N - k} \neq 1$ for $k =2, \dots, N-1$, and $(p/q)^{L -N + k} \neq 1$ for $k = 2, \dots , N-1$. Then, except for the stationary solution $z_1 = \cdots= z_N =1$, $\vec{z}$ is inadmissible if and only if $z_i = z_j$ for two distinct indices $i$ and $j$ .
\end{proposition}

\begin{rem} 
As the Bethe ansatz equations stand, we cannot have a solution to~\eqref{L} such that $\xi_k = 0 $ or $\infty$. However, in the next section, we will reformulate these equations with auxiliary maps where~\eqref{L} is not necessarily satisfied, in the sense that solutions with $\xi_k = 0 $ or $\infty$ might be included, but if we have that all the $\xi_k \neq 0$ or $\infty$, then~\eqref{L} might be satisfied. Thus, in hindsight of the following section, we add the cases that $\xi_k = 0$ or $\infty$.
\end{rem}

\begin{proof}
	First, we show that if $z_i = z_j$, then the eigenfunction $u_{\vec{z}}(\vec{x})$ is trivial. This breaks into checking two case, where $A_{\sigma}$ is well-defined on $\mathbb{C}^N$ and when it is not. Indeed, note that if $z_i = z_j = 1$ or $p/q$ we have that $(p + q z_i z_j -z_i)/(p +q z_i z_j -z_j) = 0/0$, which not well-defined on $\mathbb{C}^N$. In order to handle the case that is not well-defined, we must work on a certain blow-up of $(\mathbb{CP}^1)^N$ where we resolve the auxiliary maps $\psi$ and $\phi$ used to reformulate the Bethe ansatz equations~\eqref{L} (see sections~\ref{sec:maps} and appendix~\ref{blow}) and the $A_{\sigma}$\rq{}s become well-defined. In any case, the important fact is that terms of the form $(p + q z_i z_j -z_i)/(p +q z_i z_j -z_j) = 0/0$ are well-defined on the blow-up of $(\mathbb{CP}^1)^N$, which are defined via a sort of generalized L\rq{}H\^opital\rq{}s rule where the value is given by the way that you approach the locus where $0/0$.\par
	In the case, where $z_i = z_j \neq 1, p/q$, we let $\tau \in S_N$ be the transposition that swaps $i$ and $j$. Then we have $A_{\sigma} = -A_{\sigma \tau}$ and for all $\vec{x}$ we have
\[
u_{\vec{\xi}}(\vec{x}) = \sum_{\sigma \in S_N} A_{\sigma} \prod_{j=1}^{N} \xi_{\sigma (j)}^{x_j} = \sum_{\sigma \in S_N} A_{\sigma \tau^2} \prod_{j=1}^{N} \xi_{\sigma \tau^2(j)}^{x_j} =\sum_{\sigma' \in S_N} A_{\sigma' \tau } \prod_{j=1}^{N} \xi_{\sigma' \tau(j)}^{x_{\tau(j)}} = \sum_{\sigma' \in S_N} - A_{\sigma'} \prod_{j=1}^{N} \xi_{\sigma'(j)}^{x_j} = - u_{\vec{z}}(\vec{x}).
\]
It follows that $\vec{z}$ is inadmissible, as claimed.\par
	On the hand, when $z_i = z_j = 1, p/q$ (and all other $z$\rq{}s are distinct), we have that $A_{\sigma} = - a A_{\sigma \tau}$ for some $a \in \mathbb{CP}^1$ by the resolution using the blow-up. In particular, we have that $(p +q z_i z_j -z_i)/(p +q z_i z_j -z_j) = a$. Then, by the same procedure as above, we have that $(1+ a)u_{\vec{z}}(\vec{x}) =0$. So, if we have that $a \neq -1$, it follows that $\vec{z}$ is inadmissible. Still, if $a =-1$, we look at the $i^{th}$ Bethe equation~\eqref{L}. We have that for $z_i =z_j =1$
	\begin{align*}
	1&= (-1)^{N} \prod_{k \neq j ,i} \frac{p + q  z_k - 1}{p +q z_k - z_k} \\
	&= (-1)^{N} \prod_{k \neq j ,i } \frac{ q  (z_k - 1)}{p(1 - z_k)}\\
	&=   (q/p)^{N-2},
	\end{align*}
or we have that for $z_i = z_j = p/q$
	\begin{align*}
	\left(\frac{p}{q} \right)^{L}&= (-1)^{N} \prod_{k \neq j ,i} \frac{p + p z_k - p/q}{p +p z_k - z_k} \\
	&= (-1)^{N} \prod_{k \neq j ,i } \frac{ (p/q) (q z_k - p)}{p - q z_k}\\
	&=   (p/q)^{N-2}.
	\end{align*}
So, we have a contradiction since we assume that $(q/p)^{N-1} \neq 1$ and $(p/q)^{L-2-N} \neq 1$, meaning that $z_i = z_j =1, p/q$ (and all the other $z$\rq{}s distinct) don\rq{}t lead to non-trivial solutions of the Bethe ansatz equations.\par
	The rest of the cases we must consider are when $z_{i_1} = \cdots = z_{i_k} = 1 , p/q$. In those cases, we arrive at similar conclusion via similar methods where we either have $A_{\sigma}$\rq{}s that lead to an inadmissible solution or we arrive at a contradiction where $(q/p)^{N - k} = 1$  or $(p/q)^{L -N + k} \neq 1$ for $k =2 , \dots, N-1$. Note that, when $k =N$, there is no contradiction since $(q/p)^{N - N} = 1$, and thus we have that the stationary solution with $z_1 = \cdots = z_N =1$ is indeed an admissible solution. This establishes the forward direction of the lemma.\par
	Now, suppose that $\vec{z}$ is an inadmissible solution. We show that there must be $i \neq j$ such that $z_i = z_j$, and we do this by induction on $N$. In particular, we take an inadmissible solution $\vec{z}$ and assume, for contradiction, that all the parameters are finite, non-zero, and pairwise distinct. Then, we show that this leads to an inadmissible solution with $N-1$ parameters $z_2, \dots, z_N$, and from the $N-1$ case we must have that two of the $N-1$ parameters must the same, which contradicts our original assumption that all the parameters are distinct. Therefore, an inadmissible solution must have $i \neq j$ such that $z_i =z_j$.\par
	First, we take $N$ parameters $z_1, \dots, z_N$,  finite, non-zero, and pairwise distinct. Then, we note that the equation
	\begin{equation}\label{eq:putnam}
	c_1 z_1^{x_1} + \cdots + c_N z_N^{x_1} = 0, \hspace{5mm} x_1 \in \mathbb{Z},
	\end{equation}
where $c_i \in \mathbb{C}$, implies that all the coefficients must be zero (i.e. $c_1 =\cdots = c_N = 0$). This follows by reducing~\eqref{eq:putnam} to the same equation with $N-1$ variables instead of $N$ and then applying an induction argument. So, we divide equation~\eqref{eq:putnam} by $z_1^{x_1}$ and take the difference with the corresponding equation with $x_1 + 1$ instead of $x_1$. We have that
	\[
	c_2\left(1- \frac{z_2}{z_1}\right) (z_2/z_1)^{x_1} + \cdots + c_N\left(1 - \frac{z_N}{z_1}\right)(z_N/z_1)^{x_1} = 0.
	\]
Since all the parameters are non-zero and distinct, we have that all of the terms $z_2/z_1, \dots , z_N/z_1$ are finite, non-zero, and pairwise distinct. So, by induction, we have that all the coefficients $c_2(1- z_2/z_1), \dots, c_N(1- z_N/z_1)$ are zero, and since all the parameters $z_1, \dots, z_N$ are pairwise distinct, we have that $c_2 = \cdots = c_N =0$. Lastly, note that this reduces~\eqref{eq:putnam} to $c_1 z_1^{x_1} = 0$, which clearly implies that $c_1 =0$, establishing the base case and the fact that all the coefficients of~\eqref{eq:putnam} must be identically zero for the equation to hold if all the parameters are finite, non-zero, and pairwise distinct.\par
	At this stage, we continue with the assumption that all the parameters $z_1, \dots, z_N$ are distinct, and we rewrite $u_{\vec{z}}(\vec{x}) = 0$ making the dependence on $x_1$ explicit. We have
	\begin{equation}
	u_{\vec{\xi}}(\vec{x}) = \sum_{\sigma \in S_N} A_{\sigma} \prod_{i =1}^{N} z_{\sigma(i)}^{x_i} = \sum_{j=1}^{N} \left ( \underset{\sigma^{-1}(j)=1}{\sum_{\sigma \in S_N}} A_{\sigma} \prod_{i =2}^{N} z_{\sigma(i)}^{x_i} \right) z_j^{x_1} =0.
	\end{equation}
Then, by the argument in the previous paragraph, we have that
	\begin{equation}\label{eq:inad_ind}
	 \underset{\sigma^{-1}(j)=1}{\sum_{\sigma \in S_N}} A_{\sigma} \prod_{i =2}^{N} z_{\sigma(i)}^{x_i} =0 
	\end{equation}
for every $j =1, \dots, N$. Moreover, for $j=1$, note that $\{\sigma \in S_N \mid \sigma(1) = 1 \} \cong S_{N-1}$, and we can prove our result by induction on $N$. In particular, if we have that $u_{\vec{z}}(\vec{x}) = 0$ (where $A_{\sigma}$'s are any constants not necessarily give by~\eqref{amplitudes}) for $N-1$ implies that two of the parameters in $\vec{z}$ must be the same, then we apply that  result to~\eqref{eq:inad_ind} and we have that two of the parameters in $z_2, \dots, z_N$ must be the same, contradicting the original assumption that all the parameters $z_1, \dots, z_N$ are distinct. Therefore, there is no inadmissible solution with finite, non-zero, and pair-wise distinct parameters $\vec{z}$.\par
It only remains to show the result for the base case $N=2$. That is, if
\[
A_{12}\xi_1^{x_1}\xi_2^{x_2} + A_{21}\xi_1^{x_2}\xi_2^{x_1} = 0
\]
then $\xi_1 = \xi_2$. Recall that we showed this in the two-particle case. This establishes the lemma in the backwards direction.
\end{proof}

\subsection{Maps \texorpdfstring{$\psi$}{psi} and \texorpdfstring{$\phi$}{phi}}
\label{sec:maps}

Now that we know how to classify the inadmissible solutions, we'll count the number of admissible solutions. For example, we count all the solutions and subtract the solutions such that any two $\xi_i$ are equal, but in doing so we over count the number of inadmissible solutions as in the case when three $\xi_i$ are equal. We must then take care to do the proper count. In fact, what we want to count is the cardinality of the set $\pi^{-1} (p)$ for a generic p, and we can use the Lefschetz Theorem~\ref{lefschet}. Following Langlands and Saint-Aubin~\cite{Langlands}, we introduce $\binom{N}{2}$ variables and decompose the Bethe ansatz equations~\eqref{L} into
\begin{equation} \label{w}
\omega_{(k,l)} = - \frac{p + q\xi_k \xi_l -\xi_k}{p + q\xi_k \xi_l -\xi_l} \text{ for } 1 \leq k < l \leq N
\end{equation}
\begin{equation} \label{R}
R(\xi_i) = \left( \prod_{1 \leq j < i} \omega_{(j,i)}\right) \left( \prod_{i < j \leq N} \omega_{(i,j)}^{-1}\right) \qquad \text{ for } i =1, \ldots , N
\end{equation}%
where $R: \mathbb{CP}^1 \rightarrow \mathbb{CP}^1 $ has $L$ zeros at $\alpha_i \in \mathbb{CP}^1$ and $L$ poles at $\beta_i \in \mathbb{CP}^1$ with $\beta_i = \frac{1}{q} - \frac{p}{q \alpha_i}$ which is a generalization of the function $\xi \mapsto \xi^L$. Now, note that our equations~\eqref{w} and~\eqref{R} are over the space $\mathbb{C}^N \times \mathbb{C}^{\binom{N}{2}}$ except for the divisor where the numerator and denominator of the right hand side of~\eqref{w} vanish. In order to use the Lefschetz theorem~\ref{lefschet} we need to compactify the spaces we are working with and make the right-hand of~\eqref{w} well-defined.

As in the $N=2$ case, we will count the set of solutions to the system of equations~\eqref{w} and~\eqref{R} by counting the number of coincidences of two smooth maps with compact domain and range. Again, we start with the domain and range of our maps to be $X:=( \mathbb{CP}^1)^{\frac{N(N+1)}{2}}$ and denote the projection maps into its $\frac{N(N+1)}{2}$ factors by 
\begin{align*}
\pi_{i}: X &\rightarrow \mathbb{CP}^1 \qquad \text{for} \qquad i=1, \ldots, N \\
\pi_{(k,l)}:X &\rightarrow \mathbb{CP}^1 \qquad \text{for} \qquad 1 \leq k<l \leq N
\end{align*}
Define the rational map $\psi_N: X \dashedrightarrow X$ by 
\[
\psi_N: (\xi, \omega) \mapsto ({\xi'} , {\omega'})
\]
with 
\begin{align*}
\xi &= ([\xi_0^i: \xi_1^i])_{i=1}^N\\
{\xi}' &= ([{\xi'}_0^i : {\xi'}_1^i])_{i=1}^N\\
\omega &= ([\omega_0^{(k,l)} : \omega_1^{(k,l)}])_{1 \leq k<l\leq N} \\
{\omega}' &= ([{\omega'}_0^{(k,l)} : {\omega'}_1^{(k,l)}])_{1 \leq k<l\leq N}
\end{align*}
and
\begin{align*}
[{\xi'}_0^i : {\xi'}_1^i]&= \left[ \prod_{1 \leq j < i} \omega_0^{(j,i)} \prod_{i < j \leq N} \omega_1^{(i,j)} :  \prod_{1 \leq j < i} \omega_1^{(j,i)} \prod_{i < j \leq N} \omega_0^{(i,j)} \right] \\
[{\omega'}_0^{(k,l)} : {\omega'}_1^{(k,l)}]&= [ p \xi_0^k \xi_0^l +q \xi_1^k \xi_1^l - \xi_0^k \xi_1^l : - (p \xi_0^k \xi_0^l +q \xi_1^k \xi_1^l - \xi_0^l \xi_1^k )]
\end{align*}
Note that $\psi_N: X \dashedrightarrow X$ is well-defined away from the codimension 2 subvariety
\[
Z := \bigcup_{1 \leq k < l \leq N} \left( Z_{(k,l)}^0 \cup Z_{(k,l)}^1 \right),
\]
where $Z_{(k,l)} : = \left\{pt \in X | \pi_k (pt) = \pi_l(pt) =[2q : 1 + (-1)^i(1 - 2 p)] \right\}$. Also, we define the smooth map $\phi_N: X \rightarrow X$ by 
\[
\phi_N: (\xi, \omega) \mapsto ({\xi''}, {\omega''})
\]%
with
\begin{align*}
\xi &= ([\xi_0^i: \xi_1^i])_{i=1}^N\\
{\xi}'' &= ([{\xi''}_0^i : {\xi''}_1^i])_{i=1}^N\\
\omega &= ([\omega_0^{(k,l)} : \omega_1^{(k,l)}])_{1 \leq k<l\leq N} \\
{\omega}'' &= ([{\omega''}_0^{(k,l)} : {\omega''}_1^{(k,l)}])_{1 \leq k<l\leq N}
\end{align*}%
and
\begin{align*}
[{\xi'}_0^i : {\xi'}_1^i]&= R \left( [\xi_0^i: \xi_1^i] \right)\\
[{\omega'}_0^{(k,l)} : {\omega'}_1^{(k,l)}]&= [\omega_0^{(k,l)} : \omega_1^{(k,l)}]
\end{align*}%
where $R: \mathbb{CP}^1 \rightarrow \mathbb{CP}^1$ is the function defined by equation~\eqref{R}.

Now define a chart $U := (\xi_0^i \neq 0 , \omega_0^{(k,l)} \neq 0 \mid i =1, \ldots , N \text{ and } 1 \leq k<l \leq N )$ and $f: U \rightarrow \mathbb{C}^{\frac{N(N+1)}{2}}$ with $(\xi ,\omega) \mapsto (\vec{\xi}, \vec{\omega})$ with
\begin{align*}
\xi &= ([\xi_0^i: \xi_1^i])_{i=1}^N\\
\vec{\xi} &= (\xi_i)_{i=1}^N\\
\omega &= ([\omega_0^{(k,l)} : \omega_1^{(k,l)}])_{1 \leq k<l\leq N} \\
\vec{\omega} &= (\omega_{(k,l)})_{1 \leq k<l\leq N}
\end{align*}%
and 
\begin{align*}
\xi_i &= {\xi_1^i}/{\xi_0^i} \\
\omega &= {\omega_1^{(k,l)}}/{\omega_0^{(k,l)}}
\end{align*}%
Then, note that the system of equations~\eqref{w} and~\eqref{R} on $\mathbb{C}^{\frac{N(N+1)}{2}} - f\left(U \cap Z \right)$ is equivalent to 
\begin{equation} \label{corr3}
f \circ \psi_N \circ f^{-1} (\vec{\xi} , \vec{\omega}) = f \circ \phi_N \circ f^{-1} (\vec{\xi} , \vec{\omega}).
\end{equation}

Thus, our strategy is to count the admissible solutions of $\psi_N(p) = \phi_N(p)$ and discard the solutions that are not on $U$. We formalize this count in the next section. 
But as we have mentioned in the $N=2$ case, we first need to obtain a smooth resolution of $\psi_N : (\mathbb{P}^1)^{\frac{N(N+1)}{2}} \dashedrightarrow (\mathbb{P}^1)^{\frac{N(N+1)}{2}}$ to apply Theorem~\ref{lefschet}. We do this resolution by applying a sequence of blow-ups to our domain which results in a smooth space $C$ with a projection map $\pi: C \rightarrow X$ such that $C - \pi^{-1}(Z) \cong X- Z$. Then, we define smooth maps $\tilde{\psi}_N , \tilde{\phi}_N: C \rightarrow X$ such that these maps agree with our maps $\psi_N $ and $\phi_N$ on $C - \pi^{-1}$. We give the details of this construction in Appendix~\ref{blow}. 

Therefore, we have that the solutions of the Bethe ansatz equations (admissible and inadmissible) are the same as the solutions to the equation
\[
\tilde{\phi}_N(pt) = \tilde{\psi}_N(pt) \qquad \text{ for } pt \in C
\]

\begin{proposition} \label{directsum}
Let $\tilde{C} = \mathrm{Blow}_Z(C)$. There exist smooth maps $\tilde{\psi}, \tilde{\phi}: \tilde{C} \to X$ that agree with $\psi, \phi$ away from the exceptional divisor (i.e. away from the subvariety where we have performed the blow-up). Furthermore, if we have $H_*(\tilde{C}) = H_*(C) \oplus H$, then $\tilde{\phi}_*(H) = 0$.
\end{proposition}

\begin{proof}
See Appendix~\ref{blow} for the existence of the maps. For the second statement, it follows from the general properties of the blow-up that \[H_*(C) = \text{Im } \pi^* \oplus H\] where $\text{Im} \pi^* \cong H_*(X)$~\cite{GriffithsHarris}. Moreover, since $\tilde{\phi} = \phi \circ \pi$, we have 
\[
\tilde{\phi}^*(H_*(X)) \cap H = \emptyset. \qedhere
\]
\end{proof}

\subsection{Lefschetz theorem}

We reduced the problem of finding eigenstates of the Bethe ansatz to finding solutions of the equation
 \begin{equation} \label{correspondence}
 \phi(pt) = \psi(pt).
 \end{equation} 
In algebraic topology, one can answer how many such solutions exist by the Lefschetz theorem.

Our cohomology spaces have dimensions
\begin{align*}
\dim H_{2i} (X) &= \binom{N + \binom{N}{2}}{i} \\
\dim H_{2i+1} (X) &= 0 .
\end{align*}

To compute the trace of the map $\psi_i \phi^i$, we choose a basis of $H_{*} (X)$. In this setting, we can choose a basis of $H_{*} (X)$ as 
\[
\{ e_S \otimes e_T \mid S \subseteq [N] \text{ and } T \subseteq [N]^{(2)} \}
\] 
where $e_S \otimes e_T$ is represented by the submanifold
\[
\left( \prod_{i \in S}  \mathbb{CP}^1_i \right) \times \left( \prod_{(k,l) \in T}  \mathbb{CP}^1_{k,l} \right).
\]

We would like to compute the induced maps $\psi_i$ and $\phi^i$. Still, we have yet to determine a basis for $H_*(C)$, but we have, from the Appendix~\ref{blow} and Proposition~\ref{directsum} on blow-ups, that
\[
H_* (C) = H_*(X) \oplus H.
\]
Moreover, we have the commutative diagram

\[
\begin{tikzcd}
C \arrow{r}{\tilde{\phi}} \arrow{d}{\pi} & X \\
X \arrow{ur}[swap]{\phi}
\end{tikzcd}
\]

where $\pi: C \rightarrow X$ is the projection from the blow-up and ${\phi} : (\vec{\xi}, \vec{w}) \mapsto (\vec{\xi}^L, \vec{w})$ on an affine chart. So, we have that $\tilde{\phi}^* = \pi^* \circ \phi^*$ with $\pi^*(H_*(X)) \cap H = \emptyset$. Thus, we can factor the maps

\[
\begin{tikzcd}
H_*(X) \arrow{r}{\tilde{\phi}^*} \arrow{dr}{\phi^*} & H_*(C) \arrow{r}{\tilde{\psi}_*} & H_*(X)\\
& H_*(X) \arrow{u}[swap]{\pi^*} \arrow{ur}[swap]{\psi_*}
\end{tikzcd}
\]

This means that, when we are computing $\tilde{\psi}_*$, we only need to consider the image of the elements in $H_{*}(X) \subseteq H_{*}(C)$ where we have already established a basis. Now, let's consider some examples. For instance,
\[
1 \overset{\phi_*}{\mapsto} L^N \cdot 1.
\]
We have that the element $1 \in H_{0}(X)$ is represented by a point (i.e. $1= [pt]$). Then, 
\[
\phi^*[pt] = [\phi^{-1} (pt)] = [L^N pts] = L^N [pt] = L^N \cdot 1 .
\]
Also, we have that 
\[
e_i \overset{\phi^*}{\mapsto} L^{N-1} \cdot e_i.
\]%
Indeed, we can represent $e_i =[(1, \ldots, 1, \xi_i, 1, \ldots , 1 )]$ and we have that
\[
\phi^{-1}((1, \ldots, 1, \xi_i, 1, \ldots , 1 )) = \bigcup_{k_j = 1, \dots , L} (e^{2 \pi i k_1/L }, \ldots, \xi_i, \ldots , e^{2 \pi i k_N/L } ).
\]
In general, we have the following result.

\begin{lemma} \label{soma}
\[
e_S \otimes e_T \overset{\phi^*}{\mapsto} L^{N-|S|} \, e_S \otimes e_T
\]
\end{lemma}
Now, let's consider the $\psi_*$ map. For example, we have
\[
1 \overset{\psi_*}{\mapsto} 1
\]
since $\psi_*[pt] = [\psi (pt) ] = [pt] =1$. Also, we have that 
\[
e_i \overset{\psi_*}{\mapsto} \sum_{i < j} e_{ij} + \sum_{i > j} e_{ji}
\]%
since we have that the image of $\mathbb{CP}^1_i$ is 
\[
\left( \bigcup_{i > j} \mathbb{CP}^1_{ij} \right) \cup \left( \bigcup_{i < j} \mathbb{CP}^1_{ji} \right).
\]%
Similarly, we have that 
\[
e_{ij} \overset{\psi_*}{\mapsto} e_i + e_j
\]
since the image of $\mathbb{CP}^1_{ij}$ is $\mathbb{CP}^1_i \cup \mathbb{CP}^1_j$. Also, we consider one more instructive example where $N=3$, where we have that 
\[
e_1\otimes e_2 \otimes e_3 \otimes e_{1,2} \otimes e_{1,3} \otimes e_{23} \overset{\psi_*}{\mapsto} 0.
\]
This follows from the fact that we may represent $e_1\otimes e_2 \otimes e_3 \otimes e_{1,2} \otimes e_{1,3} \otimes e_{23} =[X]$, we have that $\dim_{\mathbb{R}} X = 12$, and $\dim_{\mathbb{R}} \psi (X) < 12$. The latter can easily be seen since the coordinates of the image are not independent. Note that
\[
(w_{12} w_{13}) (w_{12}^{-1} w_{2,3}) (w_{13}^{-1} w_{23}^{-1}) = 1.
\]
Therefore, the map $\psi_*$ is not as straightforward as $\phi^*$, but in general, this is the only situation that we will encounter that will give us a zero in our computations.


\begin{rem}
In the rest of the paper, we will fix the maps $\phi,\psi : C \rightarrow X$ and we will compute the Lefschetz number for different restrictions of the maps $\psi$ and $\phi$ to certain subvarieties $C(\mathcal{A},\mathcal{B}) \subseteq C$ and $X(\mathcal{A}, \mathcal{B}) \subseteq X$, introduced in the next section. For convenience, we will avoid the cumbersome notation $\lambda(\psi|_{C(\mathcal{A},\mathcal{B})} , \phi_{C(\mathcal{A},\mathcal{B})})$
and simply write
\[
\lambda(\mathcal{A},\mathcal{B}).
\]
\end{rem}

\section{Proof of Theorem~\ref{mainthm}}

\subsection{Computing traces by counting trees}

In this section, we show how to compute the Lefschetz numbers of the subvarieties containing inadmissible solutions by enumerating rooted forests. 

\begin{proposition}
\label{trees}
Given $\psi , \phi : C \rightarrow X$ as above, with the induced homology maps $\psi_i: H_i (C) \rightarrow H_i(X)$ and $\phi^i : H_i (X) \rightarrow H_i(C)$, we have
\[
\sum_{i=0}^{\dim X} (-1)^i \, \mathrm{Tr} (\psi_i \phi^i ) = \sum_{f\in \mathcal{F}_N} L^{n(f)},
\]
where $\mathcal{F}_N$ denotes the set of planted forests on the vertex set $[N]$ and $n(f)$ denotes the number of components of a forest $f$.
\end{proposition}

\begin{proof}
First, note that from the commutative diagram

\[
\begin{tikzcd}
C \arrow{r}{\phi} \arrow{d}{\pi} & X \\
X \arrow{ur}[swap]{\phi'}
\end{tikzcd}
\]

where $\pi$ is the projection from the blow-up and by definition $\phi = \phi' \circ \pi$. We obtain the following commutative diagram in homology

\[
\begin{tikzcd}
H_*(X) \arrow{r}{\phi^*} \arrow{dr}[swap]{(\phi')^*} & H_*(C) \arrow{r}{\psi_*} & H_*(X)\\
& H_*(X) \arrow{u}[swap]{\pi^*} \arrow{ur}[swap]{(\psi_*)|_{H_*(X)}}
\end{tikzcd}
\]

In Appendix~\ref{blow} and Proposition~\ref{directsum}, we explain that indeed $\pi_*$ is indeed an orthogonal projection. From this, we have that 
\[
\mathrm{Tr}(\psi_* \phi^*) = \mathrm{Tr}((\psi_*)|_{H_*(X)} (\phi')^*).
\]
Thus, we only consider computations on the direct summand $H_*(X)$ of $H_*(C)$, and we will abuse notation and use $ (\phi')^*$, $(\psi_*)|_{H_*(X)}$ and $\phi^*$, $\psi_*$ interchangeably, respectively.

Next, we introduce a basis for $H_*(X)$ via Kunneth's theorem. That is, given the inclusion maps of the different factors of $\mathbb{CP}^1$ in $X$
\begin{align*}
 j_i &: \mathbb{CP}^1 \hookrightarrow X \qquad \text{ for } i =1, \ldots, N \\
j_{(k,l)} &: \mathbb{CP}^1 \hookrightarrow X \qquad \text{ for } 1 \leq i < j \leq N
\end{align*}%
we define elements in $H_2(X)$
\begin{align*}
e_i &:= (j_i)_* ([\mathbb{CP}^1]) \text{ for } i =1, \ldots, N \\
e_{(k,l)} &:= (j_{(k,l)})_* ([\mathbb{CP}^1]) \qquad \text{ for } 1 \leq i < j \leq N
\end{align*}
and given $S \subseteq [N]$ and $T \subseteq [N]^{(2)}$ we define elements in $H_{2|S|}(X)$ and $H_{2|T|}(X)$, respectively

\begin{align*}
e_S & := \bigotimes_{i \in S} e_i\\
e_T & := \bigotimes_{(k,l) \in T} e_{(k,l)}.
\end{align*}

Then, by Kunneth's theorem we have a basis for $H(X)$
\[
\mathbb{B} := \{ e_S \otimes e_T \mid S \subseteq [N] \text{ and } T \subseteq [N]^{(2)} \}
\]
and note that $\deg (e_S \otimes e_T) = 2|S|+ 2|T|$.

Moreover, for convenience, we define an inner product 
\[
\langle \,\cdot\,, \,\cdot\, \rangle : H(X)^{\otimes 2} \rightarrow \mathbb{C} 
\]
such that $\mathbb{B}$ is an orthonormal basis. Then, we have that
\begin{align*}
\sum_{i=0}^{\dim_{\mathbb{R}}X} (-1)^i Tr (\psi_i \phi^i ) &= \sum_{b \in \mathbb{B}} (-1)^{\deg(b)} \langle\psi_* \phi^*(b) , b\rangle\\
&= \sum_{b \in \mathbb{B}} (-1)^{\deg(b)} \langle\psi_* (b ), \phi^*(b)\rangle
\end{align*}
where the last equality comes from the fact that $\phi^*$ is represented by a diagonal matrix under the basis $\mathbb{B}$. Also, it is clear that for every basis element we have 
\[
\phi^*(e_S \otimes e_T) = L^{N -|S|} \, e_S \otimes e_T.
\]

Thus, we concentrate on computing 
\[
\langle\psi_* (e_S \otimes e_T ), e_S \otimes e_T\rangle.
\]

First, we define $X_1 = (\mathbb{CP}^1)^N$ and $X_2 = (\mathbb{CP}^1)^{\binom{N}{2}}$ and we write $X = X_1 \times X_2$. Now, given inclusions 
\begin{align*}
j_1 &: X_1 \hookrightarrow X\\
j_2 &: X_2 \hookrightarrow X
\end{align*}
and projections
\begin{align*}
\pi_1 &: X \rightarrow X_1 \\
\pi_2 &: X\rightarrow X_2
\end{align*}
we consider
\begin{align*}
\psi_{11} &: X_1 \rightarrow X_1 \\
\psi_{12} &: X_1 \rightarrow X_2 \\
\psi_{21} &: X_2 \rightarrow X_1 \\
\psi_{22} &: X_2 \rightarrow X_2
\end{align*}%
and we note that $\psi_{11}$ and $\psi_{22}$ are constant. Thus, we have that 
\[
 \langle\psi_* (e_S \otimes e_T ), e_S \otimes e_T\rangle = \langle(\psi_{12})_*(e_S), e_T\rangle \, \langle(\psi_{21})_*(e_T) , e_S\rangle
\]

One easy consequence of this is that $\deg e_S = \deg e_T$ for the above quantity to be non-zero, which means that $|S| = |T|$.

Now, we work on another decomposition of $X$ by considering a map from the basis elements $\mathbb{B}$ and the set of graphs on $N$ labeled vertices $\mathcal{G}_N$ given by
\[
\gamma: e_S \otimes e_T \mapsto ([N], T) 
\]%
where $[N]$ is the vertex set and $T$ is the edge set. Then, fixing $T$ we have a decomposition of $([N],T)$ into 
\[
([N], T ) = \left[ \bigcup_{x=1}^r (V_x, T_x)\right] \cup (V_{r+1}, \emptyset)
\]%
where $(V_x, T_x)$, for $1 \leq x \leq r$, are connected components with at least one edge and $V_{r+1}$ is the set of vertices without edges (i.e. $V_{r+1} = [N] \setminus \bigcup_{x=1}^r V_x $). Also, we define $T_{r+1} := [N]^{(2)} \setminus T$ and the spaces
\begin{align*}
G_x &:= (\mathbb{CP}^1)^{|V_x|} \times (\mathbb{CP}^1)^{|T_x|}
\end{align*}
with the projections of $X$ onto the $i^{th}$ $\mathbb{CP}^1$ factors and $(k,l)^{th}$ $\mathbb{CP}^1$ factors such that $i \in V_x$ and $(k,l) \in T_x$
\[
p_x: X \rightarrow G_x \qquad \text{for } 1 \leq x \leq r+1.
\]

Thus, we obtain the decomposition
\[
X = \prod_{x=1}^{r+1} G_x.
\]

Next, given the inclusions
\begin{align*}
q_x &: G_x \hookrightarrow X \qquad \text{ for } 1 \leq x \leq r+1,
\end{align*}
we define the maps
\[
\psi_{x,y} := p_y \circ \psi \circ q_x : G_x \rightarrow G_y,
\]
and we note that $\psi_{x,y}$ is constant unless $x=y$. Then, given a basis element $e_S \otimes e_T$, we can write via Kunneth's formula
\[
e_S \otimes e_T = \left[ \bigotimes_{x=1}^r (e_{S_x} \otimes e_{T_x}) \right] \otimes e_{S_{r+1}}
\]
where $S_x := V_x \cap S$ for $1 \leq x \leq r+1$. Thus, from this decompostion, we obtain
\begin{align*}
\langle\psi_* (e_S \otimes e_T ), e_S \otimes e_T\rangle &= \\
\langle(\psi_{r+1,r+1})_* (e_{S_{r+1}}), e_{S_{r+1}}\rangle &\prod_{x=1}^r \langle(\psi_{x,x})_* (e_{S_x} \otimes e_{T_x} ), e_{S_x} \otimes e_{T_x}\rangle 
\end{align*}

Thus, we may assume that $([N],T)$ is connected and that 
\[
S(T) := \bigcup_{(i,j) \in T} \{i,j\} = [ N] 
\]

Also, we have that $S \subseteq [N]$ and from before we have that $|S| = |T|$. Then, $|T| \leq N$. Since $([N], T)$ is a connected graph, we must have that $|T| = N$, which means that $([N],T)$ is a connected graph with one cycle, or $|T| = N-1$, which means that $([N],T)$ is a tree. Otherwise, if $|T| < N-1$, then $([N],T)$ is not connected.

Then, given the following lemma

\begin{lemma}
Given $T \subseteq [N]^{(2)}$ such that the graph $([N] , T)$ is connected,
\[
\langle(\psi_{21})_*(e_T), e_S\rangle = \delta_{|S|,N-1} = \langle(\psi_{12})_*(e_S), e_T \rangle.
\]
\end{lemma}

We have that, if $([N],T)$ is connected, then $\langle\psi_*(e_S \otimes e_T) ,e_S \otimes e_T \rangle =1$ if and only if $|T|= |S| = N-1$. This means that $([N],T)$ is a tree and that $[N] -S = \{\rho\}$. Therefore, we have that $([N],T,[N]-S)$ is a rooted tree. So we obtain 
\begin{align*}
\sum_{i=0}^{\dim_{\mathbb{R}}X} (-1)^i Tr (\psi_i \phi^i ) &= \sum_{e_S \otimes e_T \in \mathbb{B}} (-1)^{\deg(e_S \otimes e_T)} \langle\psi_* (e_S \otimes e_T ), \phi^*(e_S \otimes e_T)\rangle \\
&=\sum_{e_S \otimes e_T \in \mathbb{B}} (-1)^{2|S| +2|T|} L^{N-|S|} \langle \psi_* (e_S \otimes e_T ), e_S \otimes e_T \rangle \\
&= \sum_{f \in \mathcal{F}} L^{\# roots} \\
&= \sum_{f \in \mathcal{F}} L^{n(f)}. \qedhere
\end{align*}
\end{proof}

Now, with this Proposition~\ref{trees} we have reduced the problem of counting ``fixed points" to a combinatorial problem of counting planted forests. Still we need to further refine the combinatorics to account for the inadmissible solutions. Indeed, the count from~\ref{trees} includes inadmissible solutions. Fortunately, we can classify inadmissible solutions using prop.~\ref{inad} and ``special" partitions of $[N]$, which we introduce.

\begin{definition}
Let $L$ and $N$ be positive integers. A \emph{labeled enhanced partition of $N$} is a pair
\[
(\mathcal{A},\mathcal{B}) = (\{A_1 , \ldots, A_s \} , \{ (B_1, \B_1, d_1), \ldots, (B_b, \B_b, d_b) \} ),
\]
where $A_1, A_2, \ldots, A_s, B_1, B_2, \ldots, B_b, \B_1, \B_2, \ldots, \B_b$ are pairwise disjoint non-empty sets whose union is $[N]$ and $d_1, d_2, \ldots, d_b \in [L]$. We use the notation $\mathbb{EP}^{l}(N)$ to denote the set of all labeled enhanced partitions of $N$.

Similarly, an \emph{enhanced partition of $N$} is a pair
\[
(\mathcal{A},\mathcal{B}) = (\{A_1 , \ldots, A_s \} , \{ (B_1, \B_1), \ldots, (B_b, \B_b) \} ),
\]
where $A_1, A_2, \ldots, A_s, B_1, B_2, \ldots, B_b, \B_1, \B_2, \ldots, \B_b$ are pairwise disjoint non-empty sets whose union is $[N]$. We use the notation $\mathbb{EP}(N)$ to denote the set of all enhanced partitions of $N$.
\end{definition}

From this, we have a ``forgetful map" $\sigma: \mathbb{EP}^l(N) \rightarrow \mathbb{EP}(N)$ defined by
\[
(\{A_1 , \ldots, A_s \} , \{ (B_1, \B_1; d_1), \ldots, (B_b, \B_b; d_b) \} ) \mapsto (\{A_1 , \ldots, A_s \} , \{ (B_1, \B_1), \ldots, (B_b, \B_b) \} )
\]
which ``forgets" the label of a partition. Then, given an enhanced partition we can define subvarieties of $X$ and $C$ as follows:
\[
C(\mathcal{A},\mathcal{B}) := \left\{ p \in C \left|
\begin{array}{l}
\tilde{\pi}_i(p) = \tilde{\pi}_j(p) \text{ if } i,j \in A \text{ for some } A \in \mathcal{A}, \\
\tilde{\pi}_k(p) = \alpha_r \text{ if } k \in B \text{ for some } (B, \B, r) \in \mathcal{B} \\
\tilde{\pi}_l(p) = \beta_r \text{ if } k \in \B \text{ for some } (B, \B, r) \in \mathcal{B} 
\end{array} \right. \right\}
\]
and 
\begin{align*}
X(\mathcal{A},\mathcal{B}) := \left\{ p \in X \left|
\begin{array}{l}
\pi_i(p) = \pi_j(p) \text{ if } i,j \in A \text{ for some } A \in \mathcal{A}, \\
\pi_k(p) = [1:0] \text{ if } k \in B \text{ for some } (B, \B, r) \in \mathcal{B} \\
\pi_l(p) = [0:1] \text{ if } k \in \B \text{ for some } (B, \B, r) \in \mathcal{B} 
\end{array} \right. \right\}
\end{align*}

One should now note that if $pt \in C(\mathcal{A}, \mathcal{B})$ for any non-trivial labeled enhanced partition $(\mathcal{A},\mathcal{B})$ and 
\[
\tilde{\phi}(pt) = \tilde{\psi}(pt)
\]%
then $pt$ is an inadmissible solution. Thus, we will use (labeled) enhanced partitions $(\mathcal{A}, \mathcal{B})$ and the corresponding subvarieties $X(\mathcal{A}, \mathcal{B})$ and $C(\mathcal{A}, \mathcal{B})$ to count inadmissible solutions. That is, we will count the number of solutions of 
\[
\tilde{\phi}|_{C(\mathcal{A}, \mathcal{B})} (pt) = \tilde{\psi}|_{C(\mathcal{A}, \mathcal{B})}(pt)
\]%
to count the number of inadmissible solutions satisfying ``conditions" $(\mathcal{A}, \mathcal{B})$.

\begin{rem}
Note that we characterize the inadmissible solutions by $z_i =z_j$, $z_i = \alpha_j, \beta_j$. It might appear that this would discount the steady state solution of $z_1 = \cdots = z_N=1$, but as we noted in the case of two particles, the steady state solution moves away from the locus of $z_i =z_j$ making our characterization of inadmissible solutions for generic $R(z)$ legitimate.
\end{rem}

First, we define the analogous set of planted forests $\mathcal{F}_N$. Given a tuple of conditions $(\mathcal{A}, \mathcal{B})$, we define $\mathcal{F}(\mathcal{A}, \mathcal{B})$ to be the set of planted forests with $\mathcal{A} \cup \mathcal{B}$ as vertices where vertices from $\mathcal{B}$ may only be roots. Moreover, we define a multiplicity for each $f \in \mathcal{F}(\mathcal{A}, \mathcal{B})$. For this, we denote as $E(f)$ the set of edges of a planted forest. We have that the root on each component of $f$ induces an orientation on each component where the root can be thought of as a source and the edges are oriented away from the root. Then, we write the edges of $f$ as tuples $(a,b)$ where the edge connects vertices $a$ and $b$, and it is oriented from $a$ to $b$. Recall, that $a$ is an $\mathcal{A}$-set or a $\mathcal{B}$-set, and denote by $\#(a)$ the cardinality of the set $a$ if $a \in A$ or the cardinality of $a(1) \cup a(2)$ if $a \in \mathcal{B}$. Now, we can define the multiplicity of a forest $f \in \mathcal{F}(\mathcal{A}, \mathcal{B})$ as 
\[
m(f):= \prod_{(a,b) \in E(f)} \#(a).
\]

 We have the following result generalizing Proposition~\ref{trees}.


\begin{lemma}
\label{two}
Given a labeled enhanced partition $(\mathcal{A}, \mathcal{B})$, we have that the number solutions of $\phi(pt) = \psi(pt)$ satisfying the inadmissibility conditions $(\mathcal{A}, \mathcal{B})$ is given by
\[
\sum_{f \in \mathcal{F}(\mathcal{A}, \mathcal{B})} m(f) \, L^{n(f)-|\mathcal{B}|}.
\]
\end{lemma}

\begin{proof}

We obtain this result from the Lefschetz Theorem~\ref{lefschet} and Proposition~\ref{trees}. The difference is that we are looking for solutions to the Lefschetz problem in subvarieties of $C$ and $X$ that are determined by the conditions $(\mathcal{A},\mathcal{B})$. Namely, consider the projection maps onto the $i^{th}$ $\mathbb{CP}^1$ factors of $X$ and $C$, respectively
\begin{align*}
\pi_i &: X \rightarrow \mathbb{CP}^1\\
\tilde{\pi}_i := \pi_i \circ \pi &: C \rightarrow \mathbb{CP}^1
\end{align*}%
where $\pi: C \rightarrow X$ is the projection from the blow-up.

Note that we have inclusion maps
\begin{align*}
j_C &: C(\mathcal{A},\mathcal{B}) \hookrightarrow C\\
j_X &: X(\mathcal{A},\mathcal{B}) \hookrightarrow X
\end{align*}%
and we define maps $\phi_{\mathcal{A},\mathcal{B}} , \psi_{\mathcal{A},\mathcal{B}} : C(\mathcal{A},\mathcal{B}) \rightarrow X(\mathcal{A},\mathcal{B})$ such that the following diagram commutes

\begin{equation}\label{com}
\begin{tikzcd}[column sep=large]
C \arrow{r}{\phi,\psi}& X\\
C(\mathcal{A} ,\mathcal{B}) \arrow[hook]{u}{j_C} \arrow{r}{\phi_{\mathcal{A}, \mathcal{B}} , \psi_{\mathcal{A}, \mathcal{B}}  }& X(\mathcal{A}, \mathcal{B}) \arrow[hook]{u}{j_X} 
\end{tikzcd}
\end{equation}

Then, by the Lefschetz Theorem~\ref{lefschet}, we have that the number of solutions (with multiplicity) that satisfy the inadmissibility conditions $(\mathcal{A},\mathcal{B})$ is given by 
\[
\lambda(\mathcal{A},\mathcal{B}) := \sum_{i=0}^{\dim_{\mathbb{R}}X(\mathcal{A},\mathcal{B})} (-1)^i Tr ((\psi_{\mathcal{A},\mathcal{B}})_i (\phi_{\mathcal{A},\mathcal{B}})^i ).
\]%
Again, we compute this number by introducing a basis of $H(X(\mathcal{A},\mathcal{B}))$ and an inner product that makes this basis orthonormal. Then, along with functorial properties of the commutative diagram~\eqref{com}, we can compute $\lambda (\mathcal{A},\mathcal{B})$ using Proposition~\ref{trees}. So, we introduce elements in $e_{A_i} \in H_{2}(X(\mathcal{A},\mathcal{B}))$ for $1 \leq i \leq s$ where 
\begin{equation} \label{e1}
(j_C)^* (e_k) = (j_X)^* (e_k) = e_{A_i} \qquad \text{ for any } k \in A_i
\end{equation}
\begin{equation} \label{e2}
(j_C)_*(e_{A_i}) = (j_X)_*(e_{A_i}) = \sum_{k \in A_i} e_k 
\end{equation}%
and for any subset $S \subseteq A$, we define 
\begin{align*}
e_S =\bigotimes_{A_i \in S} e_{A_i}.
\end{align*}

Also, note that the elements $e_{(k,l)} \in H_2(X)$ remain unaffected by the inclusion map. That is, we have $(j_X)^*(j_X)_*(e_{(k,l)}) = e_{(k,l)}$, and by abuse of notation we denote the image of $e_{(k,l)} \in H_2(X)$ under the push-forward of the inclusion map by $e_{(k,l)} \in H_{2}(X(\mathcal{A},\mathcal{B}))$. For $T \in [N]^{(2)}$, we let $e_T \in H(X(\mathcal{A},\mathcal{B}))$ be the push-forward of $e_T \in H(X)$ under the inclusion map $j_X$. Then, we have a basis for $H(X(\mathcal{A},\mathcal{B}))$ defined as follows
\[
\mathbb{B}_{\mathcal{A},\mathcal{B}} : = \{ e_S \otimes e_T \mid S \subseteq A \text{ and } T \subseteq [N]^{(2)} \}
\]%
and an inner product that makes this basis orthonormal
\[
 \langle \,\cdot\, , \,\cdot\, \rangle_{\mathcal{A},\mathcal{B}} : H(X(\mathcal{A},\mathcal{B}))^{\otimes 2} \rightarrow \mathbb{C}.
\]%
First, one can see that 
\begin{align}
\langle (\phi_{\mathcal{A},\mathcal{B}})^*(e_S \otimes e_T) , e_S\otimes e_T\rangle_{\mathcal{A},\mathcal{B}} = L^{|A| - |A \cap S|}.
\end{align}%
This follows from the fact that 
\[
\langle(j_X)^*(e_S \otimes e_T) ,(j_C)^*( e_S\otimes e_T)\rangle_{\mathcal{A},\mathcal{B}} = \langle e_S \otimes e_T, e_S\otimes e_T\rangle
\]%
and from the commutative diagram~\eqref{com}
\[
\langle(\phi_{\mathcal{A},\mathcal{B}})^* (j_X)^*(e_S \otimes e_T) ,(j_C)^*( e_S\otimes e_T)\rangle_{\mathcal{A},\mathcal{B}} = \langle\phi^* ( e_S \otimes e_T), e_S\otimes e_T\rangle.
\]

We note that the number of components is the same as the number of roots and we will show shortly that elements in $B$ and elements in the complement of $S$ in $A$ must be roots. We have that $n(f) - |\mathcal{B}| = |A| - |A \cap S|$ as we wished.

Now, we wish to compute 
\[
\langle(\psi_{\mathcal{A},\mathcal{B}})_*(e_S \otimes e_T) , e_S\otimes e_T\rangle_{\mathcal{A},\mathcal{B}}.
\]

First, note that, from~\eqref{e2}, we have 
\[
\frac{\langle(j_C)_*(e_S \otimes e_T) , (j_X)_*(e_S \otimes e_T)\rangle}{\prod_{i= i}^s |A_i|} = \langle e_S \otimes e_T , e_S \otimes e_T \rangle_{\mathcal{A},\mathcal{B}}
\]%
and so, by the commutative diagram and the functoriality of the push-forward, we have that 
\[
\langle(\psi_{\mathcal{A}, \mathcal{B}})_* (e_S \otimes e_T) , e_S \otimes e_T \rangle_{\mathcal{A},\mathcal{B}
} =\frac{\langle \psi_*(j_C)_*(e_S \otimes e_T) , (j_X)_*(e_S \otimes e_T)\rangle}{\prod_{i=1}^s |A_i|}.
\]
Again, since we know $(j_C)_*$ and $(j_X)_*$ from~\eqref{e1},~\eqref{e2}, and by Proposition~\ref{trees}, we have that
\[
\langle \psi_*(j_C)_*(e_S \otimes e_T) , (j_X)_*(e_S \otimes e_T)\rangle =1
\]%
if and only if there exist a choice $a_i \in A_i$ for $i =1, \ldots s$ such that $([N], T, [N] - \{a_1, \ldots , a_s\})$ is a planted forest. Note that there is no element of $B$ in $\{ a_1 ,\ldots, a_s\}$ since the elements in $B$ are not in the image of $(j_C)_*$, which makes elements of $B$ roots in $([N], T, [N] - \{a_1, \ldots , a_s\})$. Thus, we consider the subset $\tilde{\mathcal{F}}(\mathcal{A},\mathcal{B}) \subseteq \mathcal{F}$ of planted forest on $N$ labeled vertices where $\tilde{f} \in \tilde{\mathcal{F}}(\mathcal{A},\mathcal{B})$ if and only if there are no edges within the elements of $A_i$ for $1 \leq i \leq s$ and vertices from $B_j$ for $1 \leq j \leq r$ may only be roots. Then, the above becomes 
\[
\langle \psi_*(j_C)_*(e_S \otimes e_T) , (j_X)_*(e_S \otimes e_T)\rangle =1
\]%
if and only if there exists a choice $a_i \in A_i$ for $i =1, \ldots s$ such that $([N], T, [N] - \{a_1, \ldots , a_s\})$ is in $\tilde{\mathcal{F}}(\mathcal{A},\mathcal{B})$. Moreover, note that given a planted forest $\tilde{f} \in \tilde{\mathcal{F}}(\mathcal{A},\mathcal{B})$ we may define a planted forest $f \in \mathcal{F}(\mathcal{A},\mathcal{B})$, in the set of planted forests with $A \cup B$ as the vertex set. Indeed, for every directed edge $(k,l)$ in $\tilde{f}$, where the direction is induced from the roots, we define a directed edge $(v,w)$ such that $k \in v$ and $l \in w$ and this directed edge makes $f$ into a planted forest in $\mathcal{F}(A, B)$. We denote this map by 
\[
F: \tilde{\mathcal{F}}(\mathcal{A},\mathcal{B}) \rightarrow \mathcal{F}(\mathcal{A},\mathcal{B}).
\]%
Therefore, we have 
\begin{align*}
&\sum_{e_S \otimes e_T \in \mathbb{B}(\mathcal{A},\mathcal{B})}\langle(\psi_{\mathcal{A},\mathcal{B}})_* (e_S \otimes e_T) ,e_S \otimes e_T \rangle_{\mathcal{A},\mathcal{B}} \\
 & = \sum_{e_S \otimes e_T \in \mathbb{B}(\mathcal{A},\mathcal{B})} \frac{\langle\psi_*(j_C)_*(e_S \otimes e_T) , (j_X)_*(e_S \otimes e_T)\rangle}{\prod_{i=1}^s |A_i|} \\
&= \sum_{\tilde{f} \in \tilde{\mathcal{F}}(\mathcal{A},\mathcal{B})} \frac{1}{\prod_{i=1}^s |A_i|} \\
&= \sum_{f \in \mathcal{F}(\mathcal{A},\mathcal{B})} \sum_{\tilde{f} \in F^{-1}(f)} \frac{1}{\prod_{i=1}^s |A_i|} 
\end{align*}%
Now, one has that $|F^{-1} (f)| =\prod_{(v,w) \in E(f)} \#(v) \#(w)$ since we have $\#(v) \#(w)$ choices of edges $(k,l)$ that map to $(v,w)$ under $F$. Then,
\[
 \sum_{f \in \mathcal{F}(\mathcal{A},\mathcal{B})} \sum_{\tilde{f} \in F^{-1}(f)} \frac{1}{\prod_{i=1}^s |A_i|} = \prod_{(v,w) \in E(f)} \#(v) = m(f),
\]
since each vertex only has only one incoming edge at most and otherwise its a root. So,
\begin{align*}
\sum_{e_S \otimes e_T \in \mathbb{B}(\mathcal{A},\mathcal{B})} \langle(\psi_{\mathcal{A},\mathcal{B}})_* (e_S \otimes e_T) ,e_S \otimes e_T \rangle_{\mathcal{A},\mathcal{B}} &= \sum_{f \in \mathcal{F}(\mathcal{A},\mathcal{B})} m(f),
\end{align*}
and from~\eqref{com}, we have the desired result
\begin{align*}
\sum_{e_S \otimes e_T \in \mathbb{B}(\mathcal{A},\mathcal{B})} \langle(\psi_{\mathcal{A},\mathcal{B}})_* (\phi_{\mathcal{A},\mathcal{B}})^* (e_S \otimes e_T) ,e_S \otimes e_T \rangle_{\mathcal{A},\mathcal{B}\mathcal{A},\mathcal{B}} &= \sum_{f \in \mathcal{F}(\mathcal{A},\mathcal{B})} L^{n(f) - |\mathcal{B}|} m(f). \qedhere
\end{align*}
\end{proof}

Now, we have combinatorial equations counting the inadmissible solutions using the labeled enhanced partitions. Although this refinement is enough to obtain the proper count of admissible solutions, we refine our equations still slightly more to remove the label in the labeled enhanced partitions. Thus, for an enhanced partition $(\mathcal{A}, \mathcal{B})$, we define subvarieties of $C$ and $X$ by taking the union of the subvarieties defined by the labeled enhanced partitions which match the unlabeled enhanced partition $(\mathcal{A},\mathcal{B})$. That is,
\begin{align*}
X(\mathcal{A},\mathcal{B}) =& \bigcup_{(\tilde{\mathcal{A}},\tilde{\mathcal{B}}) \in \sigma^{-1}(\mathcal{A},\mathcal{B})} X((\tilde{\mathcal{A}},\tilde{\mathcal{B}}))\\
C(\mathcal{A},\mathcal{B}) =& \bigcup_{(\tilde{\mathcal{A}},\tilde{\mathcal{B}}) \in \sigma^{-1}(\mathcal{A},\mathcal{B})} C((\tilde{\mathcal{A}},\tilde{\mathcal{B}}))
\end{align*}

This union washes out the dependence on the labels, and we have a similar result to~\ref{two}.

\begin{lemma} \label{three}
The number of solutions to $\phi(p) = \psi(p)$ satisfying the inadmissibility conditions prescribed by the enhanced partition $(\mathcal{A}, \mathcal{B})$ is
\[
\sum_{f \in \mathcal{F}(\mathcal{A}, \mathcal{B})} m(f) \, L^{n(f)}.
\]
\end{lemma}

\begin{proof}
This follows directly from the fact that $C(\mathcal{A},\mathcal{B})$ is the disjoint union of the subvarieties $C(\tilde{\mathcal{A}},\tilde{\mathcal{B}})$ with $(\tilde{\mathcal{A}},\tilde{\mathcal{B}}) \in \sigma^{-1}(\mathcal{A},\mathcal{B})$. Also, we have that $|\sigma^{-1}(\mathcal{A},\mathcal{B})| = L^{|\mathcal{B}|}$.
\end{proof}

\subsection{The enumeration of forests}

With this, we are almost ready to perform the full count of admissible solutions to the Bethe ansatz. The general strategy is to count all solutions and then subtract the inadmissible solutions using equations from Lemma~\ref{three}. One should note that an inadmissible solution satisfies many inadmissibility conditions $(\mathcal{A},\mathcal{B})$. That is, there are multiple enhanced partitions $(\mathcal{A},\mathcal{B})$ such that $pt \in C(\mathcal{A},\mathcal{B})$. For example, if $(\mathcal{A}, \mathcal{B}) = (\{1\} , \{2\}, \ldots,\{N\}; \emptyset)$, we have $C(\mathcal{A},\mathcal{B})=C$ and so all inadmissible solutions satisfy these trivial inadmissibility conditions. Thus, we need to be careful in subtracting the number of inadmissibility conditions, making sure that we only subtract each inadmissibility condition exactly once. For this, we will introduce a combinatorial factor on enhanced partitions that will perform the count properly, but first we define a partial order on $\mathbb{EP}(N)$ that will aid in the combinatorics. For any two enhanced partitions, we say that
\[
(\mathcal{A},\mathcal{B}) \leq (\mathcal{A}',\mathcal{B}') \qquad \text{if and only if} \qquad X(\mathcal{A},\mathcal{B}) \subseteq X(\mathcal{A}',\mathcal{B}').
\]

Now, we introduce the combinatorial factor on the enhanced partitions needed for the computations. For this, we introduce the following weights on enhanced partitions.

\begin{definition}
The \emph{weight} of a set $S$ is
\[ \omega(S) = (-1)^{|S|-1}(|S|-1)! .\] Let $(\mathcal{A},\mathcal{B}) $ be an enhanced partition. Then the \emph{weight} of the tuple $(\mathcal{A},\mathcal{B})$ is \[ \omega(\mathcal{A},\mathcal{B}) = \prod_{A \in \mathcal{A}}\omega(A) \prod_{B \in \mathcal{B}} \omega(B \cup \B).
\] 
\end{definition}

\noindent Note that the weight of the tuple corresponding to $X$ is 
\[
\omega((\{1\} , \{2\}, \ldots,\{N\}; \emptyset)) = 1.
\]

\begin{lemma} \label{weightsum}
Let $(\mathcal{A}(pt),\mathcal{B}(pt))$ denote the intersection of all inadmissible varieties $(\mathcal{A},\mathcal{B})$ that contain $pt$. Then
\[
\sum_{(\mathcal{A},\mathcal{B}) \geq (\mathcal{A}(pt),\mathcal{B}(pt))} \hspace{-.8cm} \omega(\mathcal{A},\mathcal{B}) = 0
\]
unless $pt$ is admissible, in which case the sum is simply unity as noted above. 
\end{lemma}
\begin{proof}
First, we will show that we only need to prove this for the case where $(\mathcal{A}(pt),\mathcal{B}(pt)) = (A_1)$, i.e. when $\xi_1 = \cdots = \xi_N$. In this case, we are summing over all partitions of $[N]$, refinements of $\{\{1,2,\ldots, N\}\}$, in which case we will see that the definition of the weights is natural and gives us our lemma. 

Given a tuple $(\mathcal{A},\mathcal{B})$, consider it as a partition $P_{\mathcal{A},\mathcal{B}}$ of $[N]$ in which we have listed separately the numbers that belong to $B_1',\ldots,B_b'$ in a set we call $B'$. Now, if we take any refinement of $P_{\mathcal{A},\mathcal{B}}$, we can map it onto an inadmissible variety by comparing each set in the partition with $B'$:\\
1. If the set is disjoint with $B'$ or contained in $B'$, then it is an $A$ set.\\
2. If the set has nonempty intersection with $B'$ but is not contained in $B'$, then it is a $B$ set with those elements from $B'$ distributed accordingly. Note that the weight is preserved under this map.

Example: Let $(\mathcal{A},\mathcal{B}) = (  \{ \{1,2\} \} ,  \{ \{3,4\}, \{5\} \}  )$. Then $P_{\mathcal{A},\mathcal{B}} = \{ \{1,2\}, \{3,4, 5\}\}$ and $B' = \{3,4\}$. Consider a refinement such as $P' = \{ \{1\}, \{2\}, \{3\}, \{4,5\}\}$. This corresponds to $( \{\{1\}, \{2\}, \{3\}  \} ,  \{ \{4\}, \{5\}\}  )$.

Since every refinement corresponds to a unique inadmissible variety containing $(\mathcal{A}(pt),\mathcal{B}(pt))$, it only remains to prove the lemma for inadmissible varieties without roots and poles, i.e. with $A$ sets only. However, the multiplicativity of $\omega$ reduces this to the case of a single $A$ set, the case $(\mathcal{A}(pt),\mathcal{B}(pt)) = (A_1)$, $\xi_1 = \cdots = \xi_N$. 

We will proceed by induction. The base case is obvious. To get a sense of how the cancellation occurs, check $N=2$: 
\[
\omega(\{\{1,2\}\}) + \omega(\{\{1\},\{2\}\}) = -1 + 1 = 0.
\]
Let $A = \{1,2, \ldots, N\}$. Consider every partition in which $N$ is in a set by itself, e.g. $\{\{1,2,\ldots,N-1\}, \{N\}\}$. By the induction hypothesis, the sum of the weights of these vanishes since the weight of the singleton $N$ can be factored out. Now consider every partition in which $N$ is in a pair, such as $\{\{1,2,\ldots,N-2\},\{N-1,N\}\}$. Again, the sum of the weights is zero by induction since the weight of the pair including $N$ can be factored out of every term. Continuing in this way, we are left with the partitions in which $N$ is in a set of cardinality $N-1$ and the original partition. The sum of these weights is
\begin{align*}
\omega(A) &= \omega(\{1,2, \ldots, N\}) + (N-1)\omega(\{\{1\},\{2,\ldots,N\}\}) \\
&= (-1)^{N-1}(N-1)! + (N-1)(-1)^{N-2}(N-2)! = 0.
\end{align*}
\end{proof}

\begin{corollary}\label{answer}
The number of admissible solutions to the Bethe ansatz equations is
\[
\sum_{(\mathcal{A},\mathcal{B})} \omega(\mathcal{A},\mathcal{B}) \, \lambda(\mathcal{A},\mathcal{B}).
\]
\end{corollary}

\subsection{The final count} 

Now that we need to compute the final sum in Corollary~\ref{answer}, it is actually more natural to relabel the forests with permutation cycles. As in our proof of Lemma~\ref{weightsum}, we note that an enhanced partition can be considered as the usual sort of set partition with certain elements marked. Now, the weight of that partition gives the number of permutations with that cycle structure with the sign encoding the sign of the permutation. So, instead of labeling our vertices with sets from the partition, we label them with cycles. Next, we consider the family of forests with vertices labeled by cycles and edges labeled by an element of the cycle labeling the vertex from which it issues. This clearly gives us a larger collection of forests, but there is no need for the weights $m$ and $\omega$ any more. 

Consider a permutation $\pi \in S_N$ and the forests whose vertices are labelled by the cycles of $\pi$. Mark some of the elements of the cycles as roots, but not all elements in a given cycle. Now, mark each edge with an element of the cycle in the vertex closer to the root. To each of these objects, we assign the weight $(-1)^{s(\pi)} L^{c(\pi)}$, where $s$ denotes the sign of the permutation and $c$ denotes the number of disjoint cycles in the permutation. Note that $(-1)^{s(\pi)} L^{c(\pi)} = (-1)^N (-L)^{c(\pi)}$. What we'd like to show is that adding up all of these weights gives $L (L-1) (L-2) \cdots (L-N+1)$. Recall that
\[
L (L-1) (L-2) \cdots (L-N+1) = \sum_{k=1}^N s(N,k) L^k,
\]
where $s(N,k)$ is the Stirling number of the first kind with $|s(N,k)|$ giving the number of permutations in $S_N$ with exactly $k$ cycles. This implies that the forests which contribute to this sum are only those with marked elements. It remains to show that the unmarked forests cancel out. 

To this end, we construct a simple sign-reversing involution. Take the smallest unmarked element $s$. If $s$ is part of a cycle of length 1, then contract the edge adjacent to this vertex closer to the root, and slot the label s in the cycle, just after the label of the contracted edge. If $s$ is part of a cycle of length greater than 1, then extend a new edge, labelled with the element that was just before $s$ in the cycle. At the other end of this new edge, label the vertex with the 1-cycle $s$, and carry the tree of $s$-labelled edges (and its ancestors) with it. This is clearly an involution and the sign changes because the number of vertices has changed by 1. Thus, when we sum the terms $(-1)^N (-L)^{c(\pi)}$ for all of the forests, the unmarked ones cancel out to zero and those that remain give us exactly 
\[
\sum_{k=1}^N s(N,k) L^k
\]
by construction and hence the number of admissible solutions to the Bethe ansatz equations is 
\[
L (L-1) (L-2) \cdots (L-N+1) = N! \binom{L}{N}
\] 
as desired for completeness.

\section{Non-generic values of the hopping rate} \label{limiting}

We noted that the Bethe ansatz doesn't work for all $p \in \mathbb{CP}^1$. Indeed for fixed $N$ and $L$, by the Riemann-Hurwitz formula~\cite{Hartshorne}, we have that there are at most a finite number of points in $\mathbb{CP}^1$ where the ansatz doesn't work, and we call these points \emph{ramification points}. The problem is not with our method of counting but rather with the fact that, for the ramification points, certain solutions will have higher multiplicity. So despite obtaining the correct count, not all the solutions are distinct. 

Let's make this more explicit and take a $p$ that is not a ramification point. Then, for this $p$, we obtain $M := \binom{L}{N}$ distinct vectors $\vec{\xi}_1(p), \ldots , \vec{\xi}_M(p)$ such that $\{ u_{\vec{\xi}_i(p)}(\vec{x}) \mid i = 1, \ldots , M\}$ forms a complete basis of the ASEP system with eigenvalues $\{ E_{\vec{\xi}_i(p)}(\vec{x}) \mid i = 1, \ldots , M\}$, respectively. Note that we make the dependence of the $\vec{\xi}_i$ on $p$ explicit, as these vectors change as we change $p$. In fact, when we take $p$ to be a ramification point $p_r$, two or more of these vectors are the same, and the basis can no longer be complete. That is, for some $i \neq j$ we have $\vec{\xi}_i(p_r) = \vec{\xi}_j(p_r)$ and $\vec{\xi}_i(p) \neq \vec{\xi}_j(p)$ for all $p \neq p_r$ in a neighborhood of $p_r$. Thus, the ansatz is not complete at $p_r$, but all is not lost. We can use the fact that the ansatz is complete at every point in a neighborhood of $p_r$ to perform a limiting process that will complete the ansatz at $p_r$. We will show our reasoning on a toy model then apply it to our particular case.

Take a $2 \times 2$ matrix $H_t$ on $\mathbb{C}^2$, which depends on a parameter $t$, with eigenvectors $\vec{v}_t$ and $\vec{w}_t$ and eigenvalues $\lambda_t$ and $\mu_t$, respectively. That is,
\begin{align*}
H_t \vec{v}_t &= \lambda_t \vec{v}_t\\
H_t \vec{w}_t &= \mu_t \vec{w}_t
\end{align*}

If $\vec{v}_t \neq \vec{w}_t$, then we have $\{\vec{v}_t, \vec{w}_t\}$ is a basis of $\mathbb{C}^2$ and it diagonalizes the matrix $H_t$. Now, assume that $\vec{v}_t \neq \vec{w}_t$ for all $t \neq t_r$ in a neighborhood of $t_r$, and $\vec{v}_{t_r} = \vec{w}_{t_r}$. We wish to find a basis of $\mathbb{C}^2$ for $t_r$ where $\{ \vec{v}_{t_r} , \vec{w}_{t_r} \}$ is no longer a basis. For this, we note that $\{ (\vec{v}_t + \vec{w}_t) , (\vec{v}_t - \vec{w}_t) \}$ is also a basis but it no longer diagonalizes the matrix $H_t$. We have 
\begin{align*}
H_t \left[ \frac{1}{2}(\vec{v}_t + \vec{w}_t) \right] &= \frac{1}{2}(\lambda_t \vec{v}_t + \mu_t \vec{w}_t) \\
H_t \left[ \vec{v}_t - \vec{w}_t \right] &= \lambda_t \vec{v}_t - \mu_t \vec{w}_t
\end{align*}

In the limit $t \rightarrow t_r$, we get that the top equation becomes an eigenvalue equation, but the second equation becomes $0 =0$. Still, we can manipulate the second equation to obtain something interesting, such as
\begin{align*}
H_t \left[ \vec{v}_t - \vec{w}_t \right] &= \lambda_t \vec{v}_t - \mu_t \vec{w}_t \\
&= (\lambda_t - \mu_t) \vec{v}_t + \mu_t ( \vec{v}_t - \vec{w}_t) \\
\Rightarrow H_t \left[ \frac{\vec{v}_t - \vec{w}_t}{\lambda_t - \mu_t} \right] &= \left( \vec{v}_t + \mu_t \frac{ \vec{v}_t - \vec{w}_t}{\lambda_t - \mu_t}\right)
\end{align*}

Then, if we take the limit $t \rightarrow t_r$, we have 
\[
\left\{ \vec{v}_{t_r} , \frac{d\vec{v}_t/ dt - d\vec{w}_t/ dt }{ d\lambda_t/dt - d\mu_t/dt }\Big|_{t=t_r} \right\}
\]%
is a basis in which the matrix becomes
\[
H_{t_r} = \begin{bmatrix}
\lambda_{t_r} & 1 \\
0 & \lambda_{t_r}
\end{bmatrix}.
\]

We complete the Bethe ansatz the same way at the ramification points. Let $p_r$ be a ramification point where 
\begin{align*}
\vec{\xi}_{j_1}^{1} &= \cdots = \vec{\xi}_{j_{k_1}}^{1} \\
 &\vdots \\
\vec{\xi}_{j_1}^{s} &= \cdots = \vec{\xi}_{j_{k_s}}^{s}
\end{align*}%

and $\vec{\xi}_{m}^{a} \neq \vec{\xi}_{n}^{b}$ if $a \neq b$. Then, the complete basis for $p_r$ becomes

\begin{align*}
\bigcup & \left\{  \frac{du_{ \vec{\xi}_{j_2}^{1} }/dp - du_{ \vec{\xi}_{j_1}^{1} }/dp }{ dE_{ \vec{\xi}_{j_2}^{1} }/dp - dE_{ \vec{\xi}_{j_1}^{1} }/dp } , \cdots ,  \frac{du_{ \vec{\xi}_{j_{k_1}}^{1} }/dp - du_{ \vec{\xi}_{j_1}^{1} }/dp }{ dE_{ \vec{\xi}_{j_{k_1}}^{1} }/dp - dE_{ \vec{\xi}_{j_1}^{1} }/dp } \right\} \\
\bigcup & \left\{  \frac{du_{ \vec{\xi}_{j_2}^{2} }/dp - du_{ \vec{\xi}_{j_1}^{2} }/dp }{ dE_{ \vec{\xi}_{j_2}^{2} }/dp - dE_{ \vec{\xi}_{j_1}^{2} }/dp } , \cdots ,  \frac{du_{ \vec{\xi}_{j_{k_1}}^{2} }/dp - du_{ \vec{\xi}_{j_1}^{2} }/dp }{ dE_{ \vec{\xi}_{j_{k_1}}^{2} }/dp - dE_{ \vec{\xi}_{j_1}^{2} }/dp } \right\} \\
 & \vdots \\
\bigcup & \left\{  \frac{du_{ \vec{\xi}_{j_2}^{s} }/dp - du_{ \vec{\xi}_{j_1}^{s} }/dp }{ dE_{ \vec{\xi}_{j_2}^{s} }/dp - dE_{ \vec{\xi}_{j_1}^{s} }/dp } , \cdots ,  \frac{du_{ \vec{\xi}_{j_{k_1}}^{s} }/dp - du_{ \vec{\xi}_{j_1}^{s} }/dp }{ dE_{ \vec{\xi}_{j_{k_1}}^{s} }/dp - dE_{ \vec{\xi}_{j_1}^{s} }/dp } \right\}.
\end{align*}

Actually, this new basis for $p_r$ may not be complete as two or more of the new basis vectors may coincide. One should note that this is the problem we are trying to fix for the point $p_r$, but this does not mean all is lost. In fact, one should note that if we end up in the case when the new basis vectors coincide, we can keep applying the same limiting process, and this procedure terminates. Indeed, we have that the $\vec{\xi}_i(p)$'s are distinct holomorphic functions depending on $p$, and if this limiting procedure doesn't terminate, it would mean that two distinct $\vec{\xi}_i(p)$'s have the same Taylor expansion at $p_r$ making the two $\vec{\xi}_i(p)$'s equal, which is a contradiction.

\section{Acknowledgements}
Eric Brattain was supported by the National Science Foundation under grant DMS-1207995. Norman Do was supported by the Australian Research Council grant DE130100650. Axel Saenz was supported by the National Science Foundation under grant DMS-0955584, UC Davis Graduate Research Mentorship Fellowship, and UC Davis Dissertation Year Fellowship. We would like to thank Dmitry Fuchs, Motohico Mulase, Brian Osserman, Austin Shapiro, and Craig Tracy for enlightening conversations. 


\appendix

\section{On blow-ups} \label{blow}

\subsection{Properties and introduction}

The blow-up is a technique in algebraic geometry that is used to resolve singularities. The simplest example is when you have a curve in a plane that intersects itself at a point. That is, we obtain something that locally looks like a cross on the plane. As a manifold, this is singular because the cross-point is not locally Euclidean. Algebraically, the tangent space at the cross-point is 2 dimensional which is higher than the expected 1 dimension. The resolution, one of many, of this singularity is easy. Introduce a time parameter and parametrize the curve using such parameter. In order to make sure that we parametrize the curve smoothly, one needs to consider the tangent vector of the path at the cross-point and check that it varies smoothly at the cross-point. Thus, we have resolved our singularity by embedding the curve in a 3-dimensional space and separating the cross-strands, but this resolution is not the blow-up resolution.

The reader should now consider this example and note that this resolution is global, meaning that we have changed the ambient space from 2 dimensions to 3 dimensions, and more importantly that the devil lies in the tangent space of the cross-point. In contrast, the blow-up is a local resolution that only uses the topology of the normal space of the cross-point. In the case of the cross, the blow-up cuts out a neighborhood of the cross-point and glues in its place a copy of the compactified normal bundle of the cross-point in the plane. This changes the embedding locally where the ambient space remains 2-dimensional but it is no longer the plane, and it separates the cross-strands according to the tangent space.

Generally, one can blow up any variety $X$ along any closed subvariety $Z$, and topologically, this corresponds to cutting out a tubular neighborhood of $Z$ and glueing the projective normal tangent bundle in its place. We don't go into the details of this techniques as it doesn't enter our computations and instead we cite the results that show that the blow-up doesn't affect our computations. We reference the reader to~\cite{Hartshorne} for further details.

First, we note where the blow-up comes in our computations. Choosing specific affine charts, we may write the coordinates of our rational maps before blow-ups as 

\[
\psi_{(k,l)} (\vec{\xi} , \vec{w}) = \tilde{w}_{(k,l)} = -\frac{F(\xi_k, \xi_l)}{F(\xi_l ,\xi_k)}.
\]

For some polynomial in two variables, $F(x,y)$. We remarked that this is not well-defined when $F(\xi_k, \xi_l) = 0 =F(\xi_l, \xi_k)$ since $0/0$ is not well-defined. To fix this problem, note that these equations define a closed subvariety $Z \subseteq X$ where the $\psi$ function is not well-defined. Also, note that the problem comes from the fact that if you approach $Z$ along distinct paths, you will obtain different limits of $-\frac{F(\xi_k, \xi_l)}{F(\xi_l ,\xi_k)}$. For example, if we approach $Z$ along $\{ F(\xi_k, \xi_l) = 0\}$, we get the limit of the fraction to be $0$. While, if we approach $Z$ along $\{ F(\xi_l, \xi_k) = 0\}$, we get the limit of the fraction to be $\infty$. Thus, we need to pull apart all the distinct limits of the function $-\frac{F(\xi_k, \xi_l)}{F(\xi_l ,\xi_k)}$ as we did with the cross-singularity, and note that the limit depends only on the normal bundle of $Z$ (i.e. how we approach $Z$). Thus, we blow up $X$ along $Z$ to make $\psi$ well-defined and smooth. Again, we will not go through the details explicitly, but we will focus on the properties of the blow-up.

\begin{lemma}[p.~602 \cite{GriffithsHarris}]
There exist a smooth map $\pi : \mathrm{Blow}_Z(X) \rightarrow X$ such that the following properties hold

\begin{enumerate}
\item $\pi^{-1} (Z) = \mathbb{P}(T_ZX)$, the inverse image of $Z$ is diffeomorphic to the projective normal tangent bundle of $Z$.
\item $\pi : \mathrm{Blow}_Z(X) \setminus \pi^{-1} (Z) \rightarrow X \setminus Z$ is an isomorphism.
\end{enumerate}

\end{lemma}

This lemma gives us the topology of the blow-up, but we are most interested in the homology of the blow-up space as it is where we do our computations.

\begin{theorem}[Prop p.~606 \cite{GriffithsHarris}]
The induced map ring $\pi^* : H^*(X) \rightarrow H^*(\mathrm{Blow}_Z(X))$ embeds $H^*(X)$ as a direct summand of $H^*(\mathrm{Blow}_Z(X))$ (i.e. $H^*(\mathrm{Blow}_Z(X)) = H^*(X) \oplus H$)
\end{theorem}

This theorem allows us to factorize the maps we are working with and determine that our computations only depend on $H^*(X)$.

\subsection{Charts}

We need to show that a smooth resolution of the $\psi$ map exists. For this, we go into the details and describe the charts of the blow-up of a manifold along a submanifold. This is all standard material that can be found in \textit{Principle of Algebraic Geometry} by Griffiths and Harris. (We copy much of this material straight from the book for reference.)

Let $\Delta$ be an $n$-dimensional disc with holomorphic coordinates $z_1, \ldots , z_n$, and let $V \subseteq \Delta$ be the locus $z_{k+1} = \cdots = z_n =0$. Let $[l_{k+1} : \cdots : l_n]$ be the homogeneous coordinates of $\mathbb{CP}^{n-k-1}$, and let 
\[
\tilde{\Delta} \subseteq \Delta \times \mathbb{CP}^{n-k-1}
\]%
be the smooth variety defined by the relations
\[
\tilde{\Delta} = \{ (z,l) | z_i l_j = z_jl_i , k+1 \leq i,j\leq n \}
\]

The projection $\pi : \tilde{\Delta} \rightarrow \Delta$ on the first factor is clearly an isomorphism away from $V$, while the inverse image of a point $z \in V$ is a projective space $\mathbb{CP}^{n-k-1}$. The manifold $\tilde{\Delta}$ with the map $\pi : \tilde{\Delta} \rightarrow \Delta$ is called the \textit{blow-up of $\Delta$ along $V$}; the inverse image $E = \pi^{-1}(V)$ is called the \textit{exceptional divisor} of the blow-up.

$\tilde{\Delta}$ may be covered by coordinate patches 
\[
U_j = (l_j \neq 0), \qquad j = k+1 , \ldots , n
\]%
with holomorphic coordinates 
\begin{align*}
z_i &=z_i, \hspace{20mm} i=1, \ldots, k \\
z(j)_i &= \frac{l_i}{l_j} = \frac{z_i}{z_j}, \hspace{10mm} i= k+1 , \ldots , \hat{j} , \ldots n \\
z_j &= z_j \hspace{20mm} i=j
\end{align*}%
on $U_j$; the coordinates $\{ z(j)_i\}$ are Euclidean coordinates on each fiber of $\pi^{-1}(p) \cong \mathbb{CP}^{n-k-1}$ of the exceptional divisior. Moreover, one can show that the blow-up $\pi : \tilde{\Delta} \rightarrow \Delta$ is independent of the coordinates chosen. This allows us to globalize the construction. That is, given a manifold $X$ and a submanifold $Z$ of codimension $k$, we may choose a collection $\{ U _{\alpha}\}$ of discs covering $Z$ such that in each disc the subvariety $Z \cap U_{\alpha}$ may be given as the locus $(z_{k+1} = \cdots = z_n =0)$, and then blow-up each disc as above and patch them together.

Now, we want to know how the subvarieties of $Y \subseteq X$ transform as we blow-up $X$ along $Z$. We define the \textit{proper transform} $\tilde{Y} \subseteq \tilde{X}_Z$ of $Y$ in the blow-up $\tilde{X}_Z$ to be the closure in $\tilde{X}_Z$ of the inverse image 
\[
\pi^{-1}(Y-Z) = \pi^{-1}(Y) - E
\]%
of $Y$ away from the exceptional divisor $E$.

\subsection{Resolution of \texorpdfstring{$\psi_N$}{psi}}

In this section, we explicitly construct the resolution of the map $\psi_N: X \dashedrightarrow X$. We are heavily influenced by the ideas of Hironaka's Theorem on the resolution of singularities \cite{Hironaka,Hironaka2,Hauser}. 

\begin{rem}In this section, our methods don't apply for the case when the hopping rate is $p = 1/2$. This is due to the fact that the subvaritey corresponding to the points where the map $\psi$ is not well-defined, for $p=1/2$, has multiplicity greater than one, as it is apparent in~\eqref{multiplicity}. This case can be approached with techniques from \cite{Hartshorne}.
\end{rem}

First, we decompose $X = X_1 \times X_2$ as in Proposition~\ref{trees}, and $\psi_N$ also decomposes 
\[
f \times g : X_1 \times X_2 \dashedrightarrow X_2 \times X_1
\]%
where $g: X_2 \rightarrow X_1$ is actually a smooth map and $f: X_1 \dashedrightarrow X_2$ is the rational part of the $\psi_N$ map. Thus, it suffices to resolve $f: X_1 \dashedrightarrow X_2$. We do this by blowing up the domain, $X_1$, several times and then we show how $f$ extends on the blow-ups.

First, recall that 
\[
f: \left( [\xi_0^i : \xi_1^i] \right)_{i=1}^N \mapsto \left( [\omega_0^{(k,l)} : \omega_1^{(k,l)}] \right)_{1 \leq k < l \leq N}
\]%
where 
\[
[\omega_0^{(k,l)} : \omega_1^{(k,l)}] = [p \xi_0^k\xi_0^l +q\xi_1^k\xi_1^l - \xi_0^k\xi_1^l : -(p \xi_0^k\xi_0^l +q\xi_1^k\xi_1^l - \xi_0^l\xi_1^k)]
\]%
where the map is undefined on the subvarieties 
\begin{equation}\label{multiplicity}
Z_{(k,l)}^i = \left\{ pt \in X \mid \pi_k(p) = \pi_l(p) = [2q : 1+ (-1)^i (1 - 2 p)] \right\}
\end{equation}
Then, after a change of variable we may assume, for convenience, that
\[
f: \left( [\xi_0^i : \xi_1^i] \right)_{i=1}^N \mapsto \left( [\omega_0^{(k,l)} : \omega_1^{(k,l)}] \right)_{1 \leq k < l \leq N}
\]%
where 
\[
[\omega_0^{(k,l)} : \omega_1^{(k,l)}] = [a(\xi_0^k\xi_1^l + \xi_0^l \xi_1^k) + \xi_0^l \xi_1^k : a(\xi_0^k\xi_1^l + \xi_0^l \xi_1^k) + \xi_0^k \xi_1^l]
\]%
and now $f$ is undefined in the subvarieties 
\begin{align*}
Z_{(k,l)}^0 &= \left\{ \left( [\xi_0^i : \xi_1^i] \right)_{i=1}^N \mid  [\xi_0^k : \xi_1^k] =[\xi_0^l : \xi_1^l] = [0:1]\right\} \\
Z_{(k,l)}^1 &= \left\{  \left( [\xi_0^i : \xi_1^i] \right)_{i=1}^N \mid  [\xi_0^k : \xi_1^k] =[\xi_0^l : \xi_1^l] = [1:0]\right\} 
\end{align*}

Now, recall, that we have that the blow-up is a local construction. Thus, to make everything explicit we will work with affine chart and describe the blow-up on each of these charts. Of course, as we performed blow-ups, we will need more charts to describe the resulting space, and we will introduce new charts accordingly. Thus, we start by introducing charts of 
\[
(\mathbb{CP}^1)^N = \left\{ \left( [\xi_0^i: \xi_1^i] \right) | i=1, \ldots, N \right\} 
\]%
indexed by subsets $I$ of $[N]$ and given by 
\[
U_I := \left( \xi_0^i \neq 0 \right) \cap \left( \xi_1^i \neq 0 \right)
\]%
with Euclidean coordinates
\[
z^I_i = \begin{cases} \xi_1^i / \xi_0^i \text{ for } i \in I \\ \xi_0^i / \xi_1^i \text{ for } i \notin I \end{cases}.
\]%
Then, we have the commutative diagram

\[
\begin{tikzcd}
U_I \arrow[hook]{d}{} \arrow[dotted]{dr}[swap]{f^I} \arrow[dotted]{drr}{f^{I}_{(k,l)}} \\
(\mathbb{CP}^1)^N \arrow[dotted]{r}[swap]{f} & (\mathbb{CP}^1)^{\frac{N}{2}(N-1)}\arrow{r}[swap]{p_{(k,l)}} & \mathbb{CP}^1 
\end{tikzcd}
\]

where $p_{(k,l)}: (\mathbb{CP}^1)^{\frac{N}{2}(N-1)} \rightarrow \mathbb{CP}^1$ is the projection onto the $(k,l)^{th}$ factor of $(\mathbb{CP}^1)^{\frac{N}{2}(N-1)}$, and we can describe $f^I_{(k,l)}$ explicitly
\[
(z^I) \mapsto 
\begin{cases} 
[ a(z^I_k +z^I_l) + z^I_l : a(z^I_k +z^I_l) + z^I_k]  \hspace{9mm} \text{ for } k,l \in I \\
[ a(z^I_k +z^I_l) + z^I_k : a(z^I_k +z^I_l) + z^I_l]  \hspace{9mm} \text{ for } k,l \notin I \\
[ a(z^I_kz^I_l +1 ) + 1 : a(z^I_k z^I_l + 1 ) + z^I_k z^I_l]  \hspace{3mm} \text{ for } k \in I , l \notin I \\
[ a(z^I_kz^I_l +1 ) + z^I_kz^I_l : a(z^I_k z^I_l + 1 ) + 1]  \hspace{3mm} \text{ for } l \in I , k \notin I.
 \end{cases}
\]
Note that, in the last two cases, we have that $f^I_{(k,l)}$is smooth. Then, after any blow-ups the proper transform of $f^I_{(k,l)}$ will remain smooth of the last two cases. Now, we wish to define the blow-ups recursively for $0\leq n \leq N$ by 
\[
C_n := \mathrm{Blow}_{Z_{n-1}} C_{n-1}
\]%
where we explain our notation now and we start by defining
\begin{align*}
C_0 &:= (\mathbb{CP}^1)^N \\
Z_0 & := \bigcup_{I \in [N]} \left( \overline{W_0^I} \cup \overline{Y_0^I} \right)
\end{align*}%
with
\begin{align}
\label{w0}
W_0^I =
\begin{cases} 
\left\{ z_i^I = 0 | i=1, \ldots , N \right\} \text{ if } [N] \subseteq I \\
\hspace{10mm} \emptyset \hspace{20mm} \text{otherwise} 
\end{cases} \\
\label{y0}
Y_0^I =
\begin{cases} 
\left\{ z_i^I = 0 | i=1, \ldots , N \right\} \text{ if } [N] \subseteq I^c \\
\hspace{10mm} \emptyset \hspace{20mm} \text{otherwise} .
\end{cases}
\end{align}

Observe that, for some (i.e. $I \neq [N],\emptyset$), we'll have $W^I_0 = Y^I_0 = \emptyset$. Regardless, we have 
\[
\overline{W^I_0} , \overline{Y^I_0} \subseteq U^I.
\]%
Also, we have that at least one of the subvarieties $W^I_0$ or $Y^I_0$ is empty, if not both, and we define 

\begin{equation}
\label{z0}
Z^I_0 := W^I_0 \cup Y^I_0 .
\end{equation}

Thus, we have the decomposition of the blow-up $C_1$ by
\begin{align*}
C_1 &= \mathrm{Blow}_{Z_0} C_0 \\
&= \bigcup_{I\subseteq [N]} \mathrm{Blow}_{Z_0^I} U^I.
\end{align*}%
Moreover, we have that 
\[
\mathrm{Blow}_{Z_0^I} U^I = 
\begin{cases}
\left\{ \left( z^I , l^I \right) \in U^I \times \mathbb{CP}^{N-1} | z^I_i l^I_j = z^I_j l^I_i \text{ for } 1 \leq i,j \leq N \right\} \hspace{8mm} \text{ if } Z^I_0 \neq \emptyset \\
\hspace{40mm}U^I \hspace{40mm} \text{ if } Z^I_0 = \emptyset.
\end{cases}
\]%
Next, we have affine charts $V^I_i \subseteq \mathrm{Blow}_{Z_0^I} U^I$ for $i=1, \ldots, N$ given by
\[
V^I_i =
\begin{cases}
\left( l^I_i \neq 0 \right) \hspace{3mm} \text{ if } Z^I_0 \neq \emptyset \\
\hspace{4mm} U^I \qquad \text{ if } Z^I_0 =\emptyset
\end{cases}
\]%
and if $\text{ if } Z^I_0 \neq \emptyset$ we have the Euclidean coordinates for the chart $V^I_i$
\[
z^I(i)_j =
\begin{cases}
l^I_j/ l^I_i \text{ for } j \neq i\\
z_i^I \hspace{4mm}\text{ for } j =i
\end{cases}
\]%
and if $\text{ if } Z^I_0 = \emptyset$ we have the Euclidean coordinates for the chart $V^I_i$
\[
z^I(i)_j = z_i^I \text{ for all } j.
\]%
Then, we have that the $V^I_i$ for $I\subseteq [N]$ and $i =1, \ldots, N$ are the charts of $C_1$, with some redundancies. Now, we can define the blow-up recursively as follows. Given $C_n$ with affine charts $V^I_{\vec{i}}$ where $I \subseteq N$ and $\vec{i} = (i_1, \ldots, i_n)$ with each $i_t \in [N]$ and pair-wise distinct, we define

\begin{align}
\label{w}
W_{n}^{I; i_1, \ldots , i_n} =
\begin{cases}
\left\{ z(i_1, \ldots, i_n)_j =0 \mid j \in [N] - \{i_1, \ldots, i_n \} \right\} \qquad \text{ if } [N] - \{i_1, \ldots, i_n \} \subseteq I \\
\hspace{25mm} \emptyset \hspace{48mm} \text{ otherwise}
\end{cases} \\
\label{y}
Y_{n}^{I; i_1, \ldots , i_n} =
\begin{cases}
\left\{ z(i_1, \ldots, i_n)_j =0 \mid j \in [N] - \{i_1, \ldots, i_n \} \right\} \qquad \text{ if } [N] - \{i_1, \ldots, i_n \} \subseteq I^c \\
\hspace{25mm} \emptyset \hspace{48mm} \text{ otherwise}.
\end{cases}
\end{align}

Note that either $W_{n}^{I; \vec{i}}$ or $Y_{n}^{I; \vec{i}}$ is empty, if not both. Then, define 
\begin{equation}
\label{z}
Z^{I;\vec{i}}_n := W_{n}^{I; \vec{i}} \cup Y_{n}^{I; \vec{i}}
\end{equation}
and 
\[
Z_n := \bigcup_{I \subseteq [N]} \bigcup_{\vec{i}} Z^{I;\vec{i}}_n
\]%
Then, we have the inductive definition
\[
C_{n+1} := \mathrm{Blow}_{Z_{n}} C_{n}
\]%
Actually, we still have to specify the order in which we perform the blow-up along the subvarieties $Z^{I;\vec{i}}_{n}$, but we claim;

\begin{lemma}
\label{disjoint2}
\begin{align}
I \neq J \text{ or }  \{ i_1 , \ldots , i_n\} \neq \{j_1 ,\ldots , j_n \} &\Rightarrow Z^{I;\vec{i}}_n \cap Z^{J;\vec{j}}_n = \emptyset \\
\label{otherwise}
I = J \text{ and }  \{ i_1 , \ldots , i_n\} = \{j_1 ,\ldots , j_n \} &\Rightarrow Z^{I;\vec{i}}_n = Z^{J;\vec{j}}_n
\end{align}
\end{lemma}%
which says that the submanifolds are disjoint or the same, meaning that the order in which we perform the blow-ups does not make a difference.

\begin{proof}
The~\eqref{otherwise} statement is clear by blowing down the subvarieties and seeing that they agree in the blow-down.

Now, assume that 
\[
Z^{I;\vec{i}}_n \neq \emptyset \neq Z^{J;\vec{j}}_n.
\]
That is, either
\[
[N]-\vec{i} \subseteq I \qquad \text{ or } \qquad [N] - \vec{i} \subseteq I^c \\
\]%
and either
\[
[N]-\vec{j} \subseteq J \qquad \text{ or } \qquad [N] - \vec{j} \subseteq J^c .
\]%
In any case, recall that
\[
Z_{n}^{I; \vec{i}} = \left\{ z(i_1, \ldots, i_n)_j =0 \mid  j \in [N] - \{i_1, \ldots, i_n \} \right\} \subseteq V^I_{\vec{i}} = \left( l^I(\vec{i}- i_n)_{i_n} \neq 0 \right).
\]%
Assume, for simplicity, that $I=J$. Also, by~\eqref{otherwise}, we may assume that the order of the indices of $\vec{i}$ and $\vec{j}$ is such that 
\begin{align*}
i_t = j_t \qquad &\text{ for } \qquad t \leq k\\
\text{and } \qquad \{ i_{k+1} , \ldots, i_n \} & \cap \{j_{k+1} , \ldots , j_n\} = \emptyset.
\end{align*}%
Now, define for any $1 \leq t \leq n $
\[
Z_{t}^{I; \vec{i}} = \left\{ z(i_1, \ldots, i_t)_j =0 | j \in [N] - \{i_1, \ldots, i_n \} \right\} \subseteq V^I_{i_1, \ldots, i_t} 
\]%
and note that $Z_{n}^{I; \vec{i}}$ is the proper transform of $Z_{t}^{I; \vec{i}}$ under the blow-up
\[
\pi_t^n : C_n \rightarrow C_t.
\]%
Also, we have that
\[
\pi^n_t \left( Z_{n}^{I; \vec{i}} \right) \cap \pi^n_t \left( Z_{n}^{J; \vec{j}} \right) = Z_{t}^{I; i_1, \ldots i_k}.
\]%
Then, we have 
\begin{align*}
\pi^n_{k+1} \left( Z_{n}^{I; \vec{i}} \right) \cap \pi^n_{k+1} \left( Z_{n}^{J; \vec{j}} \right) &= Z_{k+1}^{I; \vec{i}} \cap Z_{k+1}^{J; \vec{j}} \\
&= Z_{k+1}^{I; i_1, \ldots, i_{k+1}} \cap Z_{k+1}^{J; j_1, \ldots , j_{k+1}}.
\end{align*}%
If $k+1 < n$, then the argument follows by induction . Thus, we may assume that, 
\[
i_1 = j_1, \ldots, i_{n-1} =j_{n-1}, \text{ and } i_n \neq j_n
\]%
Then, for any $pt \in \overline{Z_n^{I;\vec{i}}}$ we have
\begin{align*}
z^I(\vec{i})_{j_n}(pt) = 0 & \Rightarrow l^I(\vec{i})_{j_n}(pt) = 0\\
& \Rightarrow pt \notin V^I_{\vec{j}} \\
& \Rightarrow pt \notin \overline{Z_n^{J;\vec{j}}}.
\end{align*}%
Therefore, we must have that 
\[
\overline{Z_n^{I;\vec{i}}} \cap \overline{Z_n^{J;\vec{j}}} = \emptyset.
\]%
Now, we need to consider the case when $I \neq J$. First, note that if $i \in I \cap J^c$, then we have
\begin{align*}
pt \in \pi_0^n\left(\overline{Z_n^{I;\vec{i}}}\right) & \Rightarrow z^I_{i_t}(pt) =0 \\
& \Rightarrow \xi^{i_t}_1(pt) = 0 \\
& \Rightarrow pt \notin U_I \\
& \Rightarrow pt \notin \pi_0^n\left(\overline{Z_n^{J;\vec{j}}}\right)\\
& \Rightarrow \pi_0^n\left(\overline{Z_n^{I;\vec{i}}}\right) \cap \pi_0^n\left(\overline{Z_n^{J;\vec{j}}}\right) = \emptyset \\
& \Rightarrow \overline{Z_n^{I;\vec{i}}} \cap \overline{Z_n^{J;\vec{j}}} = \emptyset.
\end{align*}

Thus, if we want a non-empty intersection, we must have that $I \cap J^c = \emptyset$, and by symmetry of the argument, we also have $I^c \cap J = \emptyset$. That is, we must have $I =J$ if we want a non-empty intersection, which is as we assumed in the earlier argument.
\end{proof}

From this, we can write down the blow-up explicitly
\begin{align*}
\mathrm{Blow}_{Z_n} C_n &= \bigcup_{I \subseteq [N] } \bigcup_{\vec{i}} \mathrm{Blow}_{Z_n^{I;\vec{i}}} V^I_{\vec{i}}
\end{align*}%
with
\[
\mathrm{Blow}_{Z_n^{I;\vec{i}}} V^I_{\vec{i}} = \left\{ (z^I(\vec{i}) , l^I(\vec{i}) ) \in \mathbb{C}^N \times \mathbb{CP}^{N-n-1} | z^I(\vec{i})_i l^I(\vec{i})_j = z^I(\vec{i})_j l^I(\vec{i})_i \text{ with } i,j \in [N] - \vec{i} \right\}
\]%
if $Z_n^{I;\vec{i}}$ is non-empty and otherwise we have 
\[
\mathrm{Blow}_{Z_n^{I;\vec{i}}} V^I_{\vec{i}} = V^I_{\vec{i}}.
\]%
Then, we can write affine charts $V^I_{\vec{i}+i_{n+1}}$ of $\mathrm{Blow}_{Z_n^{I;\vec{i}}} V^I_{\vec{i}}$ with $i_{n+1} \in [N] - \vec{i}$ and 
\[
V^I_{\vec{i}+i_{n+1}} = \left( l^I(\vec{i})_{i_{n+1}} \neq 0 \right) \subseteq V^I_{\vec{i}}
\]%
with Euclidean coordinates 
\begin{align*}
z^I(\vec{i}+i_{n+1})_j &= l^I(\vec{i})_j / l^I(\vec{i})_{i_{n+1}} \qquad \text{ for } j \in [N] -\{ i_1, \ldots, i_{n+1}\} \\
z^I(\vec{i}+i_{n+1})_j &= z^I(\vec{i})_j \hspace{20mm} \text{ for } j \in \{ i_1, \ldots, i_{n+1}\}
\end{align*}%
$Z_n^{I;\vec{i}}$ is non-empty and otherwise we have 
\[
V^I_{\vec{i}+i_{n+1}} = V^I_{\vec{i}}
\]%
with Euclidean coordinates 
\[
z^I(\vec{i}+i_{n+1})_j = z^I(\vec{i})_j  \text{ for all } j \in \{ i_1, \ldots, i_{n+1}\}.
\]%
Thus, from 
\[
\left( C_n , \left\{V^I_{i_1, \ldots, i_n} \right\}_{I;i_1, \ldots, i_n}\right)
\]%
we construct 
\[
\left( C_{n+1} , \left\{ V^I_{i_1, \ldots, i_{n+1}} \right\}_{I;i_1, \ldots, i_{n+1}} \right)
\]
where $C_{n+1}$ is the blow-up of $C_n$, and this concludes our explicit construction of the blow-up for $0 \leq n \leq N$. Now, we wish to show that these blow-ups resolve the map $f : (\mathbb{CP}^1)^N \dashedrightarrow (\mathbb{CP}^1)^{\frac{N}{2}(N-1)}$. Thus, we consider

\[
\begin{tikzcd}
C_n \arrow[hook]{r}{} \arrow[hook]{d}{} \arrow[dotted]{dr}[swap]{f^n}  & V_{\vec{i}}^I \arrow[dotted]{dr}{f^{I;\vec{i}}_{(k,l)}} \\
C_0 \arrow[dotted]{r}[swap]{f} & (\mathbb{CP}^1)^{\frac{N}{2}(N-1)}\arrow{r}[swap]{p_{(k,l)}} & \mathbb{CP}^1 
\end{tikzcd}
\]

where we write 
\[
f^{I;\vec{i}}_{(k,l)} : = f^n|_{V^I_{\vec{i}}} \circ p_{(k,l)}.
\]
Now, note that if $k, l \notin \vec{i}$, we have 
\begin{equation*}
f^{I;\vec{i}}_{(k,l)} : \left( z(\vec{i})_j \right)_{j=1}^N \mapsto
\end{equation*}
\begin{equation}\label{eqn}
\begin{cases}
\left[ a \left( z^I(\vec{i})_k + z^I(\vec{i})_l \right) + z^I(\vec{i})_k : a \left( z^I(\vec{i})_k + z^I(\vec{i})_l \right) + z^I(\vec{i})_l \right] \text{ if } k,l \in I^c\\
\left[ a \left( z^I(\vec{i})_k + z^I(\vec{i})_l \right) + z^I(\vec{i})_l : a \left( z^I(\vec{i})_k + z^I(\vec{i})_l \right) + z^I(\vec{i})_k \right] \text{ if } k,l \in I\\
\left[ a \left( z^I(\vec{i})_k z^I(\vec{i})_l + 1 \right) + 1 : a \left( z^I(\vec{i})_k z^I(\vec{i})_l +1 \right) + z^I(\vec{i})_l z^I(\vec{i})_k \right] \text{ if } l \in I^c, k \in I\\
\left[ a \left( z^I(\vec{i})_k z^I(\vec{i})_l + 1 \right) + z^I(\vec{i})_k z^I(\vec{i})_l : a \left( z^I(\vec{i})_k z^I(\vec{i})_l +1 \right) + 1 \right] \text{ if } l \in I, k \in I^c\\
\end{cases}
\end{equation}
where the first two cases are singular and the other two are smooth, and we note that blow-ups along $Z^{I; \vec{i} +i_k}_{n+1}$ and $Z^{I; \vec{i} +i_l}_{n+1}$ resolve the map~\eqref{eqn}. Indeed, when we blow-up $V^I_{\vec{i}}$ along $Z^{I; \vec{i} +i_k}_{n+1}$ we get that the map, in the singular cases, is given in the local coordinates by 
\begin{equation*}
f^{I;\vec{i} + i_k}_{(k,l)} : \left( z(\vec{i}+i_k)_j \right)_{j=1}^N \mapsto
\end{equation*}
\begin{equation*}
\begin{cases}
\big[ a \left( z^I(\vec{i}+ i_k)_k + z^I(\vec{i}+ i_k)_l z^I(\vec{i}+ i_k)_k \right) + z^I(\vec{i}+ i_k)_k \\: a \left( z^I(\vec{i}+ i_k)_k + z^I(\vec{i}+ i_k)_l z^I(\vec{i}+ i_k)_k \right) + z^I(\vec{i}+ i_k)_l z^I(\vec{i}+ i_k)_k \big] \text{ if } k,l \in I^c\\
\big[ a \left( z^I(\vec{i}+ i_k)_k + z^I(\vec{i}+ i_k)_l z^I(\vec{i}+ i_k)_k \right) + z^I(\vec{i}+ i_k)_l z^I(\vec{i}+ i_k)_k \\: a \left( z^I(\vec{i}+ i_k)_k + z^I(\vec{i}+ i_k)_l z^I(\vec{i}+ i_k)_k \right) + z^I(\vec{i}+ i_k)_k \big] \text{ if } k,l \in I.
\end{cases}
\end{equation*}

One can check that in both cases this map only has removable singularities at the points where we have the equality $\left(z^I(\vec{i}+ i_k)_l = z^I(\vec{i}+ i_k )_k = 0 \right)$, which can be resolved by factoring out a $z^I(\vec{i}+ i_k)_k$ term from both entries. Now, if we blow up along $Z^{I; \vec{i} +i_l}_{n+1}$, we get a similar resolution where the indices $i_k$ and $i_l$ are switched. One should note that we didn't describe what happened to the maps under the blow-up in the cases where the maps are smooth. This is because by properties of the blow-up the map will remain smooth, which is what we need. Actually, once we have a smooth resolution of a map, it will stay smooth under further blow-ups. Moreover, note that for $n = N$, $\{i_1, \dots, i_n \} = \{1, \dots, N \}$. Thus, $f^n \circ p_{(k,l)}$ for every pair $(k,l)$  and we have that 
\[
f^N: C_N \rightarrow (\mathbb{P}^1)^{\frac{N(N-1)}{2}}
\]
is smooth. This completes the proof.

\nocite{*}


\end{document}